\documentclass[12pt]{article}

\usepackage{amsthm,amsmath,mathrsfs,amsfonts,amssymb,mathtools,nicefrac}
\usepackage{bbm}
\usepackage{tikz}
\usetikzlibrary{calc}

\usepackage{graphicx,caption,subcaption}

\usepackage{natbib}
\usepackage{hyperref}
\hypersetup{
    breaklinks=true,
    colorlinks,
    citecolor=blue,
    filecolor=black,
    linkcolor=blue,
    urlcolor=black
}
\usepackage{url}

\usepackage[multiple]{footmisc} \usepackage[capposition=top]{floatrow}

\usepackage{enumerate}

\newcommand{\argmin}{\operatorname*{argmin}}

\newtheorem{lemma}{Lemma}

\theoremstyle{definition}
\newtheorem{example}{Example}
\newtheorem{remark}{Remark}
\usepackage[margin=1.25in]{geometry}
\usepackage{setspace}
\usepackage{lmodern}

\begin{document}

\def\spacingset#1{\renewcommand{\baselinestretch}{#1}\small\normalsize} \spacingset{1}

\title{\vspace{-1cm}{\bf  \Large Orthogonal Moments in Likelihood Models}\thanks{Manuscript prepared for the Journal of Applied Econometrics lecture delivered at the annual conference of the International Association for Applied Econometrics in Lisbon, Portugal, in June 2026. We thank our discussants Iv\'an Fern\'andez-Val and Chris Muris for insightful comments. \newline Funded by the European Union (ERC-NETWORK-101044319) and by the French Government and the French National Research Agency under the Investissements d'Avenir program (ANR-17-EURE-0010). Views and opinions expressed are however those of the authors only and do not necessarily reflect those of the European Union or the European Research Council. Neither the European Union nor the granting authority can be held responsible for them. }}
\author{ St\'ephane Bonhomme\thanks{sbonhomme@uchicago.edu} \\ {\small Department of Economics, University of Chicago} \and  Koen Jochmans\thanks{koen.jochmans@tse-fr.eu} \\  {\small Toulouse School of Economics, Universit\'e Toulouse Capitole} \and Martin Weidner\thanks{m.weidner@ucl.ac.uk} \\{\small Department of Economics, University College London}}
\date{\small September 2026}
\maketitle

\vspace{-.5cm}
\begin{abstract}
\noindent Many models, such as fixed-effect models for panel or network data, are hard to estimate because they feature nuisance parameters that are both numerous and estimated imprecisely. This, in general, causes an incidental-parameter problem in the estimator of the parameters of interest. The problem can be alleviated by working with an estimating equation whose expectation is insensitive to the value of the nuisance parameters. We discuss and contrast three notions of insensitivity, also called \emph{orthogonality}, in the context of likelihood models: Neyman orthogonality,  Neyman orthogonality to order $q$, and full orthogonality. Orthogonal moments are obtained by projecting the estimating equation on nested subspaces, which are spanned by, respectively, the scores of the nuisance parameters, the first $q$ derivatives of the likelihood ratio with respect to the nuisance parameters, and all likelihood ratios of the model. We give explicit constructions in binary-choice, count-data, and nonlinear regression models.

\medskip


\end{abstract}

\onehalfspacing

\renewcommand{\theequation}{\arabic{section}.\arabic{equation}}

\section{Introduction\label{sec_intro}}
\setcounter{equation}{0}

In many estimation problems the parameters can be divided into two groups. The first group contains the target parameter; it answers the question of interest. The second group consists of nuisance parameters; they are not of primary interest but have to be estimated along the way. In program evaluation where the target is an average treatment effect the nuisance parameters may describe a propensity score. In structural settings they may be the parameters of the underlying model needed to generate policy counterfactuals. In the panel and network settings that we will focus on here they are fixed effects, one or more of them for each individual, firm, or network node. 

Identification of the target parameter is typically based on a moment condition in which the nuisance parameters appear. A feasible estimating equation then requires these to be replaced by an estimator. When nuisance parameters are numerous or imprecisely estimated, as in panel and network data models with fixed effects, the estimating equation is generally biased and inference may be substantially affected.

A popular response is to replace or modify the moment function to render it less sensitive to the nuisance parameters, in a well-defined sense. In this paper we review and connect three such possibilities. The first one is \emph{Neyman orthogonality}, introduced in \citet{Neyman1959}. It requires the first derivative of the expected moment with respect to the nuisance parameters to be zero at the truth. This zero expected-Jacobian condition is at the heart of the modern double/debiased machine-learning (DDML) literature as in \cite{ChernozhukovChetverikovDemirerDufloHansenNeweyRobins2018}. It introduces enough insensitivity to permit the target parameter to be estimated at the parametric rate while allowing for (sample-split) first-step estimators that converge at a rate faster than $n^{-\nicefrac{1}{4}}$, where $n$ is the sample size (see, e.g., \citealp{Newey1994}).

In this paper we focus on models that have a likelihood structure, where the density of the data depends on a fixed-dimensional parameter vector, a vector of (many) nuisance  parameters, and possibly a set of exogenous covariates. In this setting a Neyman-orthogonal moment can be constructed by projecting the original moment function onto the orthogonal complement of the score of the nuisance parameters, as \citet{Neyman1959} originally proposed. The DDML approach to estimation then consists of generalized method of moments (GMM) estimation based on such a Neyman-orthogonal moment, with a sample-split estimate of the nuisance parameter plugged in.

However, the faster-than-$n^{-\nicefrac{1}{4}}$ convergence may be too strong a requirement in practice, when the nuisance parameters are too many or too hard to estimate. As discussed in \citet{bonhomme2024neyman}, standard panel data settings with $N$ units, $T$ periods, and $n=NT$ observations are a case in point. Indeed, unit fixed effects are estimated at the $T^{-\nicefrac{1}{2}}$ rate, which is not asymptotically negligible relative to $n^{-\nicefrac{1}{4}}=(NT)^{-\nicefrac{1}{4}}$ unless $N$ is negligible relative to $T$; that is, unless the panel is so long that there is no incidental parameter problem to start with.

In panel data, it is sometimes possible to difference out the fixed effects entirely, as happens in linear panel data models (e.g., \citealp{arellano2003panel}, \citealp{chamberlain1992efficiency}) and certain classes of nonlinear models where quasi-differencing approaches are available (e.g., \citealp{arellano2001panel}). This gives rise to moment functions whose expectation is equal to zero irrespective of the value of the nuisance parameters, and to estimators that are consistent for fixed $T$ as $N$ tends to infinity without the need for first-step estimates of nuisance parameters or sample splitting. In panel data, informative such moments do not necessarily exist (\citealp{chamberlain2010binary}). A general approach to find them is functional differencing, as introduced in \citet{bonhomme2012functional}.

Motivated by these differencing approaches in panel data settings, the second notion of insensitivity that we consider here is \emph{full orthogonality}, in the sense that the expected moment is zero at every value of the nuisance parameters. We show that the projection approach of functional differencing can be applied to general nuisance parameters, not only fixed effects in panel data. This broadens the observation in \citet{bonhomme2024functional}, who consider functional differencing in network settings. In contrast to Neyman orthogonality, full orthogonality is a global property. We show that fully orthogonal moments can be obtained by projecting the original moment function, similarly to Neyman's original proposal, with the difference that the projection is orthogonal to the closed span of likelihood ratios under the model, hence infinite-dimensional in general.

The third notion, \emph{Neyman orthogonality to order $q$} \citep{MackeySyrgkanisZadik2018}, lies between the other two. Here, all derivatives to order $q$ of the moment vanish at true parameter values. This notion is a local one, and it coincides with usual Neyman orthogonality when $q=1$. While estimation using order-$q$ orthogonal moments requires first-step estimates of nuisance parameters and sample splitting, the rate requirement is less restrictive compared to Neyman orthogonality: a rate faster than $n^{-\nicefrac{1}{2(q+1)}}$ is sufficient to ensure asymptotic unbiasedness. For example, in panel data settings where unit fixed effects are estimated at the $T^{-\nicefrac{1}{2}}$ rate, the condition requires $N$ to be negligible relative to $T^q$. This is less restrictive the larger $q$ is, and it corrects for first-order incidental parameter bias as soon as $q\geq 2$. This property is similar to large-$N,T$ bias reduction approaches in the panel literature (see, e.g., \citealp{HahnNewey2004} and \citealp{DhaeneJochmans2015b,DhaeneJochmans2015a}).

An explicit construction of such order-$q$ orthogonal moment functions is given in \citet{bonhomme2024neyman}. It consists in projecting the original moment function orthogonally to the span of the first $q$ elements of the \cite{Bhattacharyya1946} basis (which are generalized score functions). This construction generalizes the projected-score approach (\citealp{SmallMcLeish1989}, \citealp{WatermanLindsay1996}) to general moments and target parameters, and it can be interpreted as constructing a higher-order influence function (\citealp{RobinsLiTchetgenTchetgenvanderVaart2008}) in models with a likelihood structure.

\begin{figure}
\caption{Achieving  orthogonality by projection\label{fig_graph}}
\begin{center}
	\resizebox{0.98\textwidth}{!}{%
		\begin{tikzpicture}[line join=round]

			\begin{scope}[xshift=0cm]
				\node[font=\bfseries] at (0.75,4.3) {\small Neyman orthogonality};
				\draw[teal!60!black,line width=3pt,opacity=0.85] (-1.2,-0.36) -- (2.7,0.81);
				\node[teal!45!black] at (2.45,1.18) {$\mathcal S$};
				\draw[->,black!55,line width=1pt] (0,0) -- (1.80,0.54);
				\draw[->,blue!70!black,line width=2pt] (0,0) -- (1.0,3.2);
				\node[blue!70!black] at (1.2,3.4) {$u$};
				\draw[->,orange!90!black,line width=2pt,dashed] (1.80,0.54) -- (1.0,3.2);
				\node[orange!85!black] at (1.62,1.95) {$u^{*}$};
				\fill (0,0) circle (1.6pt);
				\node at (0.3,-1.45) {\footnotesize$\mathcal S={\operatorname{span}}\{\nabla_{\eta}\ln f\}$};
						\end{scope}

			\begin{scope}[xshift=6.5cm]
				\node[font=\bfseries] at (0.45,4.3) {\small order-$q$ orthogonality};
				\fill[teal!18] (-0.6,-0.55)--(2.4,0.35)--(1.5,1.95)--(-1.5,1.05)--cycle;
				\draw[teal!55!black,line width=1pt] (-0.6,-0.55)--(2.4,0.35)--(1.5,1.95)--(-1.5,1.05)--cycle;
				\node[teal!45!black] at (-0.55,0.45) {$\mathcal S$};
				\draw[->,black!55,line width=1pt] (0,0) -- (1.0,1.10);
				\draw[->,blue!70!black,line width=2pt] (0,0) -- (1.0,3.2);
				\node[blue!70!black] at (1.2,3.4) {$u$};
				\draw[->,orange!90!black,line width=2pt,dashed] (1.0,1.10) -- (1.0,3.2);
				\node[orange!85!black] at (1.32,2.15) {$u^{*}$};
				\fill (0,0) circle (1.6pt);
				\node at (0.3,-1.45) {\footnotesize$\mathcal S={\operatorname{span}}\{\nabla_{\eta}^{j} f/f,\text{ for } j=1,...,q\}$};
			\end{scope}

			\begin{scope}[xshift=13cm]
				\node[font=\bfseries] at (0.45,4.3) {\small Full orthogonality};
				\fill[teal!18] (-1.1,-0.75)--(3.0,0.55)--(2.0,2.65)--(-2.1,1.35)--cycle;
				\draw[teal!55!black,line width=1pt] (-1.1,-0.75)--(3.0,0.55)--(2.0,2.65)--(-2.1,1.35)--cycle;
				\node[teal!45!black] at (-1.05,0.55) {$\mathcal S$};
				\draw[->,black!55,line width=1pt] (0,0) -- (1.0,1.95);
				\draw[->,blue!70!black,line width=2pt] (0,0) -- (1.0,3.2);
				\node[blue!70!black] at (1.2,3.4) {$u$};
				\draw[->,orange!90!black,line width=2pt,dashed] (1.0,1.95) -- (1.0,3.2);
				\node[orange!85!black] at (1.32,2.6) {$u^{*}$};
				\fill (0,0) circle (1.6pt);
				\node at (0.0,-1.45) {\footnotesize$\mathcal S=\overline{\operatorname{span}}\{f_{\eta'}/f_{\eta},\text{ for all }\eta'\}$};
			\end{scope}
			
	\end{tikzpicture}}
\end{center}
\vskip 0.2cm
\begin{flushleft}
	{\footnotesize \textit{Notes:} $u$ denotes the original moment, $u^*$ the orthogonal moment, ${\cal{S}}$ is the space that defines the projection: scores with respect to the nuisance parameter $\eta$ for Neyman orthogonality, derivatives of likelihood ratios (or Bhattacharyya basis) for order-$q$ orthogonality, and likelihood ratios for full orthogonality. All these concepts will be defined in detail in the text.}
\end{flushleft}
\end{figure}
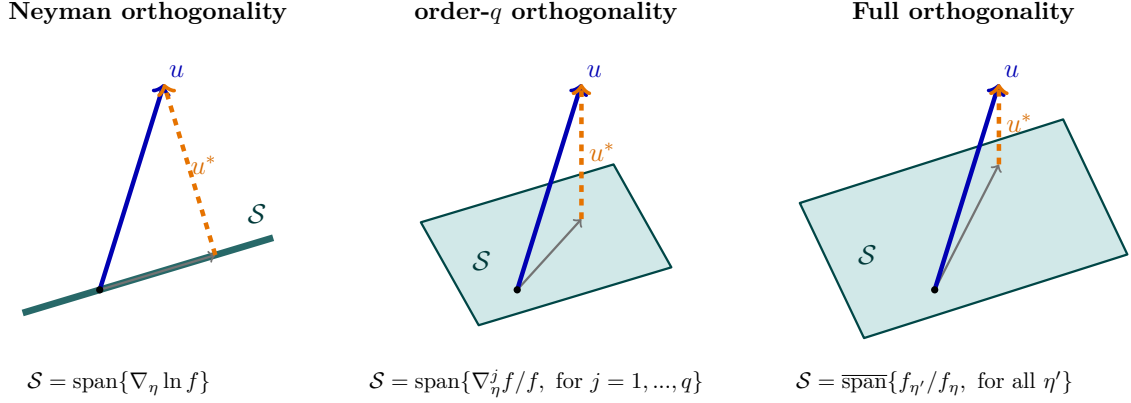

In our likelihood setting, each of the three notions of orthogonality arises from the same operation. In each case, we project the original moment function $u$ on a subspace ${\cal{S}}$ spanned by functions of the data and keep the residual $u^*$, which is the orthogonal moment. We provide a graphical illustration in Figure \ref{fig_graph}. For Neyman orthogonality, the subspace ${\cal{S}}$ is the span of the nuisance score. For orthogonality to order $q$ it is the span of the first $q$ generalized score functions. For full orthogonality it is the closed span of all likelihood ratios of the model. The three subspaces are nested, and so the trade-off between each of the three forms of orthogonality is easy to see: the larger the subspace, the more of the original moment is removed. This increases robustness to estimation noise in the nuisance parameters but, at the same time, decreases the information about the target parameter that remains present. This geometric principle allows us to connect seemingly-different procedures that have appeared in the literature. 

While our aim is largely expository, several of the connections we draw appear to be new. We characterize full orthogonality as orthogonality to the closed span of the model's likelihood ratios, which yields a projection-based construction of fully-orthogonal moments in general likelihood models that generalizes the scope of the functional differencing approach. We show that this construction coincides with Neyman's projection device, applied to a semiparametric mixture model in which the distribution of the nuisance parameter is itself the nuisance parameter. We also show that, under an analyticity condition, a moment that is orthogonal to every order $q$ is fully orthogonal, so that the three notions can be viewed as a single increasing sequence of projections. Finally, we show that the orthogonality notions and projections can be naturally adapted to handle more general target parameters such as average effects, for example.

Throughout the presentation we use several examples: a binary-choice logit model, a panel data Poisson model, and a nonlinear regression model with Gaussian errors. The logit case is sufficiently general to cover the classic fixed-effect panel-data model going back to \cite{rasch1961general}, the $\beta$-model of network formation (\citealp{bmodel}, \citealp{graham2017econometric}), as well as generic high-dimensional logistic regression \citep{sur2019modern}. An interesting feature of this model is that all projections can be written down in closed form. For example, Neyman orthogonalization amounts to variance-weighted demeaning, whereas higher-order orthogonalization consists in projecting on powers of the sufficient statistic. For the panel logit in particular, the construction naturally comes to an end when $q$ is equal to the length of the panel, in which case it delivers the conditional-logit score of \citet{rasch1961general}, \citet{andersen1970asymptotic}, and \citet{chamberlain1980analysis}. The panel Poisson example shows a case where Neyman orthogonality automatically achieves full orthogonality. Lastly, the nonlinear regression example differs from the other two in that it belongs to a \emph{curved} exponential family (\citealp{efron1975defining}), and does not possess a low-dimensional sufficient statistic. Yet, higher-order projections remain tractable in this case.

\section{Three notions of orthogonality\label{sec_notions}}
\setcounter{equation}{0}

\subsection{Moment conditions and plug-in bias\label{sec_plugin}}

We will work in a setting where we are interested in a target parameter $\theta_0\in\Theta$, known to be identified through a vector of $d_u$ moment conditions
\begin{equation}
	\mathbb{E}\left[u(y,x,\theta_0,\eta_0)\right]=0.
	\label{eq_mom}
\end{equation}
Here the vectors $y$ and $x$ are data and represent outcomes and regressors, respectively, $\eta_0\in{\cal{H}}\subseteq\mathbb{R}^{d_{\eta}}$ is a nuisance parameter, and the function $u$ is known. Estimation problems of this kind are common. For example, a propensity score or a conditional mean of potential outcomes plays the role of $\eta_0$ in program evaluation, where an average treatment effect may be the target $\theta_0$. Difference-in-differences settings, both linear and nonlinear, collect unit and time effects into $\eta_0$ and a policy parameter into $\theta_0$. In structural analysis $\eta_0$ would usually contain the parameters of the underlying structural model, whereas $\theta_0$ would then be a policy counterfactual of interest. In applications, estimands that depend on $\eta_0$, such as average treatment effects or welfare quantities, are also of interest, and we will study them in Section \ref{sec_average}.

Consider a sample of independent observations $(y_i,x_i)$, and let $\widehat\eta$ be an estimator of the nuisance parameter. It is then natural to estimate $\theta_0$ by the \emph{plug-in} generalized method of moments (GMM) estimator \citep{hansen1982large} 
\begin{equation}
	\widehat\theta= \underset{\theta\in\Theta}{\argmin}\, \left\|\sum_i u(y_i,x_i,\theta,\widehat\eta)\right\|^2_{\Omega},
	\label{eq_GMM}
\end{equation}
for a chosen positive-definite weight matrix $\Omega$, where $\|a\|^2_{\Omega}=a^{\top}\Omega a$. We denote the total number of observations as $n$.\footnote{In $N\times T$ panel-data settings an  observation is the set of $T$ time-series observations on a given unit, out of a total of $N$ units. Thus, in that case (\ref{eq_GMM}) concerns $N$ summands, whereas the total sample size is $n=NT$. In a single network involving $N$ nodes the sum collapses to a single term, whereas the number of observations is $n=N(N-1)/2$ (in the undirected case). All rate conditions below are written in terms of $n$.}

However, estimation noise in $\widehat{\eta}$ introduces bias in the estimating equation for $\theta_0$, in that
\begin{equation*}
	\mathbb{E}\left[u(y,x,\theta_0,\widehat{\eta})\right]
    \neq 0.
\end{equation*}
This bias can be substantial when $\widehat\eta$ is very imprecise; for example, this occurs when $\eta_0$ is high-dimensional as in models with fixed effects. Bias in the moment equation generally implies that the plug-in estimator $\widehat\theta$ is biased as well. This, then, invalidates conventional inference procedures based on it.

One way to proceed is to work with a modified version of the moment function $u$, say $u^*$, which still satisfies
\begin{equation}
	\mathbb{E}\left[u^*(y,x,\theta_0,\eta_0)\right]=0,
	\label{eq_mom_star}
\end{equation}
but which at the same time is less sensitive to estimation noise in the nuisance parameter. Then, estimation is simply based on (\ref{eq_GMM}) using $u^*$ instead of $u$. We will consider three notions of insensitivity, or \emph{orthogonality}.

\subsection{Definitions\label{sec_three}}

Each of the three notions of orthogonality that we consider delivers a modified estimating equation. We will denote these as $u_1^*$ for Neyman orthogonality, $u_q^*$ for orthogonality to order $q$, and $u^*_\infty$ for full orthogonality. In addition to \eqref{eq_mom_star}, these functions are designed to satisfy, respectively,
\begin{equation}
	\mathbb{E}\left[\nabla_{\eta}\, u_1^*(y,x,\theta_0,\eta_0)\right]=0,
	\label{eq_neyman1}
\end{equation}
where $\nabla_\eta u$ is the $d_{\eta}\times d_u$ matrix of partial derivatives in $\eta$,
\begin{equation}
	\mathbb{E}\left[\nabla_{\eta}^{j}\, u_q^*(y,x,\theta_0,\eta_0)\right]=0 \qquad j=1,...,q,
	\label{eq_neymanq}
\end{equation}
where $\nabla^{j}_\eta u $ denotes the matrix of partial derivatives at order $j$, again with $d_u$ columns, and, finally,
\begin{equation}
	\mathbb{E}\left[ u_{\infty}^*(y,x,\theta_0,\eta)\right]=0\quad \text{ for all }\eta\in\mathcal{H}.
	\label{eq_full}
\end{equation}
Neyman orthogonality, in (\ref{eq_neyman1}), is the weakest of the three requirements, and general recipes exist to transform a given $u$ into $u_1^*$. Full orthogonality, in (\ref{eq_full}), in contrast is the strongest. However, informative fully-orthogonal moments do not always exist. Orthogonality to a certain order $q$, in (\ref{eq_neymanq}), then interpolates between these two cases. We will study the three notions of orthogonality in Sections \ref{sec_neyman}, \ref{sec_full}, and \ref{sec_higher}, respectively, within the scope of the likelihood framework, as in \citeauthor{Neyman1959}'s \citeyearpar{Neyman1959} original contribution. We describe this framework next.

\subsection{Model and examples\label{sec_ex}}

In the conditional likelihood framework, the distribution of $y$ given $x$ belongs to the family  $$\lbrace f(y\,|\, x;\theta,\eta) : \theta\in \Theta, \eta\in\mathcal{H} \rbrace$$ and so is known up to values of the parameters $\theta$ and $\eta$. We will write $\mathbb{E}_{x;\theta,\eta}$ to denote expectations under $f(\cdot \,|\, x;\theta,\eta)$ and, as before, indicate true parameter values by $\theta_0$ and $\eta_0$. Here, the parameter of interest is $\theta_0$ and one  natural option for the function $u$ would be the score function $\nabla_\theta \ln f$, where we use $\nabla_\theta$ to denote the partial-derivative operator with respect to $\theta$. Other choices for $u$ are possible, for example when the score is difficult to calculate.

We now introduce the three examples we use as illustrations: a logistic regression model, a Poisson model with fixed effects, and a nonlinear regression model.

\begin{example}[Logit]\label{example_1} Consider a generic logistic regression 
\begin{equation}
	y = \boldsymbol{1}\left\{x_1\theta_0 + x_2\eta_0 + \varepsilon \geq 0\right\},\qquad \varepsilon\mid x_1,x_2\ \sim\ \text{i.i.d. Logistic},
	\label{eq_plrlogit}
\end{equation}
where $y\in\{0,1\}^n$, $x_1\in\mathbb{R}^{n\times d_\theta}$, $x_2\in\mathbb{R}^{n\times d_\eta}$, and the indicator acts coordinate-wise. We focus on situations where $d_{\theta}$ is moderate and $d_{\eta}$ is large, and here our target parameter is $\theta_0$. A useful property of the model is that $s=x_{2}^{\top}y$ is a sufficient statistic for $\eta_0$ at a given $\theta_0$. Model (\ref{eq_plrlogit}) nests several popular specifications in the literature, as we now illustrate. 

A first special case of model (\ref{eq_plrlogit}) is the panel binary-choice model. Given a random sample on individuals $i=1,\ldots,N$, followed over time periods $t=1,\ldots,T$, the panel logit model with a scalar outcome and individual-specific fixed effects is
\begin{equation}
	y_{it}=\boldsymbol{1}\left\{x_{it}^{\top}\theta_0+\eta_{i0}+\varepsilon_{it}\geq 0\right\},\qquad \varepsilon_{it}\mid x_{i1},\ldots,x_{iT}\ \sim\ \text{i.i.d.~Logistic}.
	\label{eq_logit}
\end{equation}
In this model, the maximum likelihood estimator is inconsistent for $T$ fixed as $N$ tends to infinity (\citealp{chamberlain1980analysis}). Likewise, the plug-in estimator $\widehat{\theta}$ based on first-step estimates $\widehat{\eta}_i$ is inconsistent too. 

A second special case is a logistic model of network formation. Consider a process of conditionally-independent undirected link formation where each pair of agents $i<j$ decide whether to form a link according to the binary-choice model
\begin{equation}
	y_{ij}=\boldsymbol{1}\left\{x_{ij}^{\top}\theta_0+\eta_{i0}+\eta_{j0}+\varepsilon_{ij}\geq 0\right\},\qquad \varepsilon_{ij}\mid x\ \sim\ \text{i.i.d.~Logistic}.
	\label{eq_netlogit}
\end{equation}
Homophily is captured by the pair-level covariates gathered in $x=(x_{ij})_{i<j}$, and degree heterogeneity by the agent effects gathered in $\eta_0=(\eta_{10},\ldots,\eta_{N0})^{\top}$, where there are $N$ agents and $n=N(N-1)/2$ pairs. The maximum likelihood estimator $\widehat{\theta}$ is consistent as $N$ tends to infinity but the asymptotic distribution of $\sqrt{n}(\widehat{\theta}-\theta_0)$ is not centered at zero (\citealp{graham2017econometric}, \citealp{Jochmans2018}). 

More generally, model (\ref{eq_plrlogit}) also nests logistic regression models with a high-dimensional matrix $x_2$ of controls. In such settings, maximum likelihood is unreliable, or fails to exist altogether, when $d_\eta$ is not small relative to $n$, as shown by \citet{sur2019modern}. Models of binary choice on bipartite networks also share this structure \citep{bonhomme2024functional}. 

\end{example}

\begin{example}[Poisson with fixed effects]\label{example_2}
For a panel of counts, a common specification is the Poisson model (see, e.g., \citealt{hausman1984econometric,blundell2002individual}). For units $i=1,\ldots,N$ and time periods $t=1,\ldots,T$ we specify
\begin{equation}
	y_{it}\mid x_{i1},\ldots, x_{iT}\ \sim\ \text{independent Poisson}\left(\mu_{it}\right),
	\label{eq_poisson}
\end{equation}
where $\mu_{it} = \exp(x_{it}^{\top}\theta_0+\eta_{i0})$ is the conditional mean. This model has the special property that the Neyman-orthogonal score is in fact fully orthogonal, as we will see in Section \ref{sec_neyman}.
\end{example}

\begin{example}[Nonlinear regression]\label{example_3}
Our third example is a nonlinear model for a continuous outcome $y\in\mathbb{R}^n$,
\begin{equation}
	y=m(x,\beta_0,\eta_0) +\sigma_0 \varepsilon,\qquad \varepsilon\mid x\ \sim\ \text{i.i.d. standard Gaussian},
	\label{eq_nonlin_AKM}
\end{equation}
for a known function $m$, a matrix of covariates $x$, $\eta_0$ a high-dimensional vector of parameters, and $\theta_0=(\beta_0^{\top},\sigma_0^2)^{\top}$ a low-dimensional parameter. Model (\ref{eq_nonlin_AKM}) contains as special cases the model with two-way (e.g., worker and firm) heterogeneity of \citet{AbowdKramarzMargolis1999} and its nonlinear generalization introduced in \citet{crippa2025identification}, under the assumption that errors are Gaussian. The model also nests models of team production (as in \citealp{AhmadpoodJones2019}, for example). Unlike the first two examples, the nonlinear regression model is not part of the standard exponential family in general, and it does not admit a low-dimensional sufficient statistic for $\eta_0$.

\end{example}

\section{Neyman orthogonality\label{sec_neyman}}
\setcounter{equation}{0}

\subsection{Motivation\label{sec_rationale}}

Recall that a function $u^*_1$ is Neyman orthogonal if two conditions hold: it has mean zero at the truth, as in \eqref{eq_mom_star}, and its derivative in $\eta$ has mean zero at the truth, as in \eqref{eq_neyman1}. To motivate this second condition, suppose for the presentation that $d_{\eta}=d_u=1$, so both the nuisance parameter and moment condition are scalar. Expand $u_1^*$ around the nuisance parameter, take expectations, and exploit \eqref{eq_mom_star} to arrive at 
\begin{align}
	\mathbb{E}\left[u^*_1(y,x,\theta_0,\widehat\eta)\right]
    = 
    \mathbb{E}\left[\nabla_{\eta}u^*_1(y,x,\theta_0,\eta_0)\,(\widehat\eta-\eta_0)\right]+(\text{higher-order terms}),
	\label{eq_expansion}
\end{align}
where here expectations are taken at $\theta_0$, $\eta_0$, and the population distribution of $x$.

We will assume that the estimator $\widehat{\eta}$ is independent of the estimation sample. Under random sampling, a way to enforce independence is to estimate $\widehat{\eta}$ on one part of the sample, and to estimate the target parameter $\widehat{\theta}$, given the first-step estimate $\widehat{\eta}$, based on the left-out observations. This \emph{sample-splitting} strategy is popular in the literature, and variability can be reduced by using multiple splits and averaging the resulting estimates, through cross-fitting \citep{ChernozhukovChetverikovDemirerDufloHansenNeweyRobins2018}.  

The dominant term on the right-hand side of (\ref{eq_expansion}) can then be written as
$$
\underset{\text{=0 by independence (sample-splitting)}}{\underbrace{\mathrm{cov}
\left[\nabla_{\eta}u^*_1(y,x,\theta_0,\eta_0) , (\widehat\eta-\eta_0)\right]}}
+
\underset{\text{=0 under \eqref{eq_neyman1}}}{\underbrace{\mathbb{E}\left[\nabla_{\eta}u^*_1(y,x,\theta_0,\eta_0)\right]}}
\,
\mathbb{E}\left[(\widehat\eta-\eta_0)\right],
$$
where the first term is zero by independence and the second term is also zero since $u_1^*$ is Neyman orthogonal, i.e., \eqref{eq_neyman1} holds.

As such, Neyman orthogonality removes the leading term from the bias expansion  (\ref{eq_expansion}). The remaining bias is of the order of the mean-squared error $\mathbb{E}(\lVert \widehat{\eta}-\eta_0 \rVert^2)$ in general. It will be asymptotically negligible provided that it goes to zero faster than $n^{-\nicefrac{1}{2}}$. This requires the first-step estimator $\widehat{\eta}$ to converge at a rate faster than $n^{-\nicefrac{1}{4}}$. Hence, with Neyman orthogonality and sample splitting, bias is asymptotically negligible even though nuisance parameters are estimated at a slower-than-parametric rate.

\subsection{\citeauthor{Neyman1959}'s \citeyearpar{Neyman1959} orthogonality  construction\label{sec_neyman_constr}}

In our likelihood setting, the requirement of Neyman orthogonality can be expressed in terms of the correlation between $u_1^*$ and the score for the nuisance parameter. To see this, we first adapt the Neyman orthogonality property to the current setting, by requiring that it holds conditionally on $x$ and at all values of the parameters. That is, we call $u_1^*$ Neyman orthogonal if 
\begin{equation}
\mathbb{E}_{x;\theta,\eta}\left[\nabla_{\eta}u^*_1(y,x,\theta,\eta)\right]=0 \qquad \text{for all }\theta\in\Theta,\ \eta\in{\cal{H}},\ \text{almost every }x.\label{eq_neyman1_lik}
\end{equation}
We will also require $u_1^*$ to be unbiased conditionally on $x$ and at all values of the parameters; that is, 
\begin{equation}
\mathbb{E}_{x;\theta,\eta}\left[u^*_1(y,x,\theta,\eta)\right]=0 \qquad \text{for all }\theta\in\Theta,\ \eta\in{\cal{H}},\ \text{almost every }x.\label{eq_mom_star_lik}
\end{equation}
The construction below will guarantee (\ref{eq_mom_star_lik}) and (\ref{eq_neyman1_lik}), which are slightly stronger than (\ref{eq_mom_star}) and (\ref{eq_neyman1}), respectively.

To proceed, note that (\ref{eq_mom_star_lik}) is equivalent to
\begin{equation}
	\int u^*_1(y,x,\theta,\eta)\,f(y\mid x;\theta,\eta)\,dy=0 \qquad \text{for all }\theta\in\Theta,\ \eta\in{\cal{H}},\ \text{almost every }x.
	\label{eq_identity}
\end{equation}
Differentiation with respect to $\eta$ produces
\begin{equation*}
	\mathbb{E}_{x;\theta,\eta}\left[\nabla_{\eta}u^*_1(y,x,\theta,\eta)\right]+\mathbb{E}_{x;\theta,\eta}\left[\nabla_{\eta}\ln f(y\mid x;\theta,\eta)\,u^*_1(y,x,\theta,\eta)^{\top}\right]=0.
\end{equation*}
Hence, letting
$$
w_1(y,x,\theta,\eta) = \nabla_\eta \ln f(y\,|\, x;\theta,\eta),
$$
it follows that the Neyman orthogonality property (\ref{eq_neyman1_lik}) is equivalent, under (\ref{eq_mom_star_lik}), to
\begin{equation}
	\mathbb{E}_{x;\theta,\eta}\left[u^*_1(y,x,\theta,\eta)\,w_1(y, x,\theta,\eta)^{\top}\right]=0\qquad \text{for all }\theta\in\Theta,\ \eta\in{\cal{H}},\ \text{almost every }x.
	\label{eq_neyman_score}
\end{equation}

The characterization in (\ref{eq_neyman_score}) readily points to the practical construction of Neyman-orthogonal moments, as initially proposed in \citet{Neyman1959}. Denote the orthogonal projection of a function $v$ onto a space ${\cal{S}}$ by
\begin{equation}\operatorname{Proj}_{x;\theta,\eta}\left[v\,\middle\|\, {\cal{S}}\right]=\underset{h\in {\cal{S}}}{\mbox{argmin}}\, \mathbb{E}_{x;\theta,\eta}\left[\left\|v(y)-h(y)\right\|^2\right].\label{eq_proj}
\end{equation}
The orthogonal moment $u_1^*$ is simply constructed by projecting the original $u$ onto the span of $w_1$ and keeping the residual, 
\begin{equation}
	u^*_1(\cdot,x,\theta,\eta)=u(\cdot,x,\theta,\eta)-\operatorname{Proj}_{x;\theta,\eta}\left[u(\cdot,x,\theta,\eta)\,\middle\|\, \operatorname{span}\left\{w_1(\cdot , x,\theta,\eta)\right\}\right].
	\label{eq_neyman_proj}
\end{equation}
The projection guarantees that (\ref{eq_neyman_score}) holds. Moreover, since $\mathbb{E}_{x;\theta,\eta}\left[w_1(y,x,\theta,\eta)\right]=0$, (\ref{eq_mom_star_lik}) holds as well whenever the original moment function $u$ is unbiased at all parameter values.

Finally, since $\eta$ is finite-dimensional, the projection is finite-dimensional as well and the orthogonal moment $u_1^*$ in (\ref{eq_neyman_proj}) takes the explicit form
\begin{equation}
	u^*_1(y,x,\theta,\eta)=u(y,x,\theta,\eta)-\Lambda_1(x;\theta,\eta)\, w_1(y, x ,\theta,\eta),
	\label{eq_neyman_construction}
\end{equation}
with least-squares projection coefficient
\begin{equation*}
	\Lambda_1(x;\theta,\eta)=\mathbb{E}_{x;\theta,\eta}\left[u(y,x,\theta,\eta)\, w_1(y,x,\theta,\eta)^{\top}\right] \, \mathbb{E}_{x;\theta,\eta}\left[w_1(y,x,\theta,\eta)\, w_1(y,x,\theta,\eta)^{\top}\right]^{-1},
\end{equation*}
provided the Gram matrix $\mathbb{E}_{x;\theta,\eta}\left[w_1(y,x,\theta,\eta)\, w_1(y,x,\theta,\eta)^{\top}\right]$ is non-singular.

\subsection{Examples\label{sec_neyman_ex}}

\paragraph{Example \ref{example_2} (continued).}

Take a single unit from the panel Poisson model (\ref{eq_poisson}) and let $\mu_t= e^{x_t^{\top}\theta+\eta} = e^{x_t^{\top}\theta} \, e^{\eta}$. Here, taking the original moment function $u$ to be the score with respect to $\theta$, we have
\begin{equation*}
	w_1 = \nabla_{\eta}\ln f=\sum_{t=1}^{T}(y_t-\mu_t),\qquad u=\nabla_{\theta}\ln f=\sum_{t=1}^{T}x_t\,(y_t-\mu_t).
\end{equation*}
By the mean-variance equality of the Poisson distribution, the least-squares projection coefficient is 
$$
\Lambda_1=\frac{\sum_t \mu_t x_t}{\sum_t \mu_t} =
\sum_t p_t x_t = \overline{x}, \qquad p_t = 
\frac{e^{x_t^{\top}\theta}}{\sum_{\tau} e^{x_{\tau}^{\top}\theta}} ,
$$ 
which is a weighted mean of the regressors. Hence, 
\begin{equation}
	u^*_1=\sum_{t=1}^{T}\left(x_t-\overline{x}\right)(y_t-\mu_t),
	\label{eq_poisson_orth}
\end{equation}
in which the weighted covariate mean $\overline{x}$ is free of $\eta$. 

In fact, in this model $u_1^*$ itself is free of $\eta$. To see this, note that
\begin{align}
	u^*_1&=\sum_{t=1}^{T}\left(x_t-\overline{x}\right)(y_t-\mu_t)\notag\\
	&=\sum_{t=1}^{T}\left(x_t-\overline{x}\right)y_t-e^{\eta}\sum_{t=1}^{T}\left(x_t-\overline{x}\right)e^{x_t^{\top}\theta}\notag\\
	&=\sum_{t=1}^{T}\left(x_t-\overline{x}\right)y_t-e^{\eta}\left(\sum_{\tau=1}^{T}e^{x_\tau^{\top}\theta}\right)\sum_{t=1}^{T}p_t\left(x_t-\overline{x}\right)\notag\\
	&=\sum_{t=1}^{T}\left(x_t-\overline{x}\right)y_t,\label{eq_poisson_free}
\end{align}
where the last equality uses that $\sum_t p_t(x_t-\overline{x})=0$. Consequently, for every $\eta'$,
\begin{equation}
	\mathbb{E}_{x;\theta,\eta}\left[u^*_1(y,x,\theta,\eta')\right]
	=\sum_{t=1}^{T}\left(x_t-\overline{x}\right)\mu_t
	=e^{\eta}\left(\sum_{\tau=1}^{T}e^{x_\tau^{\top}\theta}\right)\sum_{t=1}^{T}p_t\left(x_t-\overline{x}\right)=0,
	\label{eq_u1star_0}
\end{equation}
and the Poisson model's Neyman-orthogonal score is fully orthogonal in the sense of \eqref{eq_full}.

This robustness property arises from multiplicative separability, and relates to the property that the Poisson fixed-effects maximum likelihood estimator is immune to the
 incidental parameter problem (\citealp{hausman1984econometric}, \citealp{blundell2002individual}, \citealp{Lancaster2002}, \citealp{DhaeneJochmans2015b}).

Figure \ref{fig_mc_poisson} illustrates this observation, based on a Monte Carlo simulation exercise with $N=1000$ and $T=5$. There is a single binary covariate, and the parameter $\theta$ is shown on the x-axis. On the y-axis we report the mean (in solid), 5th and 95th percentiles (in dashed) based on $1000$ simulations. Details on the design and estimation, including sample splitting, can be found in Appendix \ref{app_simu}. 

The left graph shows the plug-in estimator. The plug-in moment is not orthogonal, and whether the resulting estimator is biased therefore depends on the first-step estimator. In our design, the bias arises because units with no positive split-sample count are dropped, which makes $e^{\widehat\eta_i}$ an upward-biased estimate of $e^{\eta_i}$ among retained units. While the bias is moderate, it is somewhat larger at lower values of $\theta_0$. In contrast, the Neyman-orthogonal estimator shown in the right graph does not depend on $\widehat\eta$ at all, and is unaffected. 

\begin{figure}[h!]
	\centering
    \begin{tabular}{cc}
    Plug-in & Neyman\\
	\includegraphics[width=.5\textwidth]{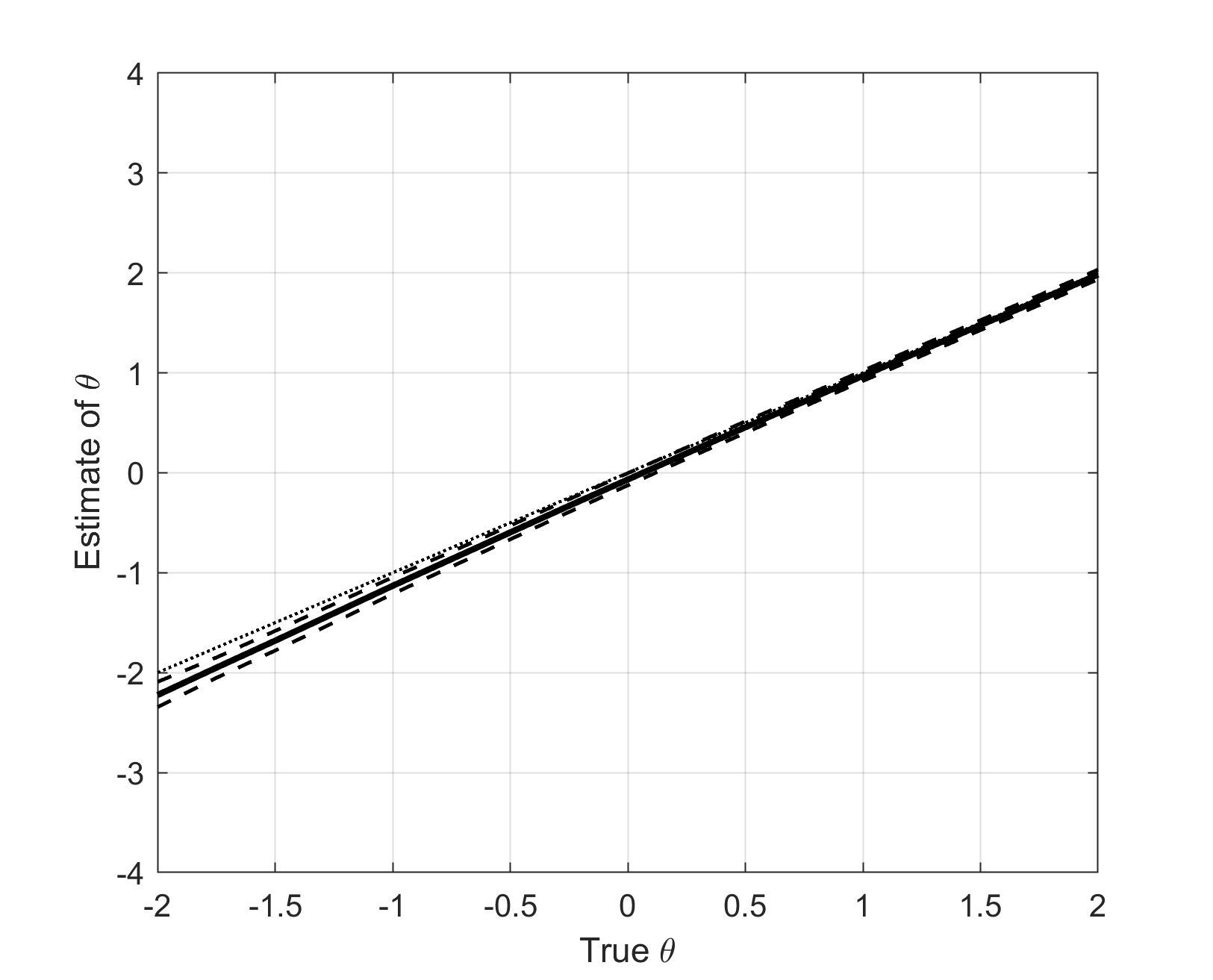}&\includegraphics[width=.5\textwidth]{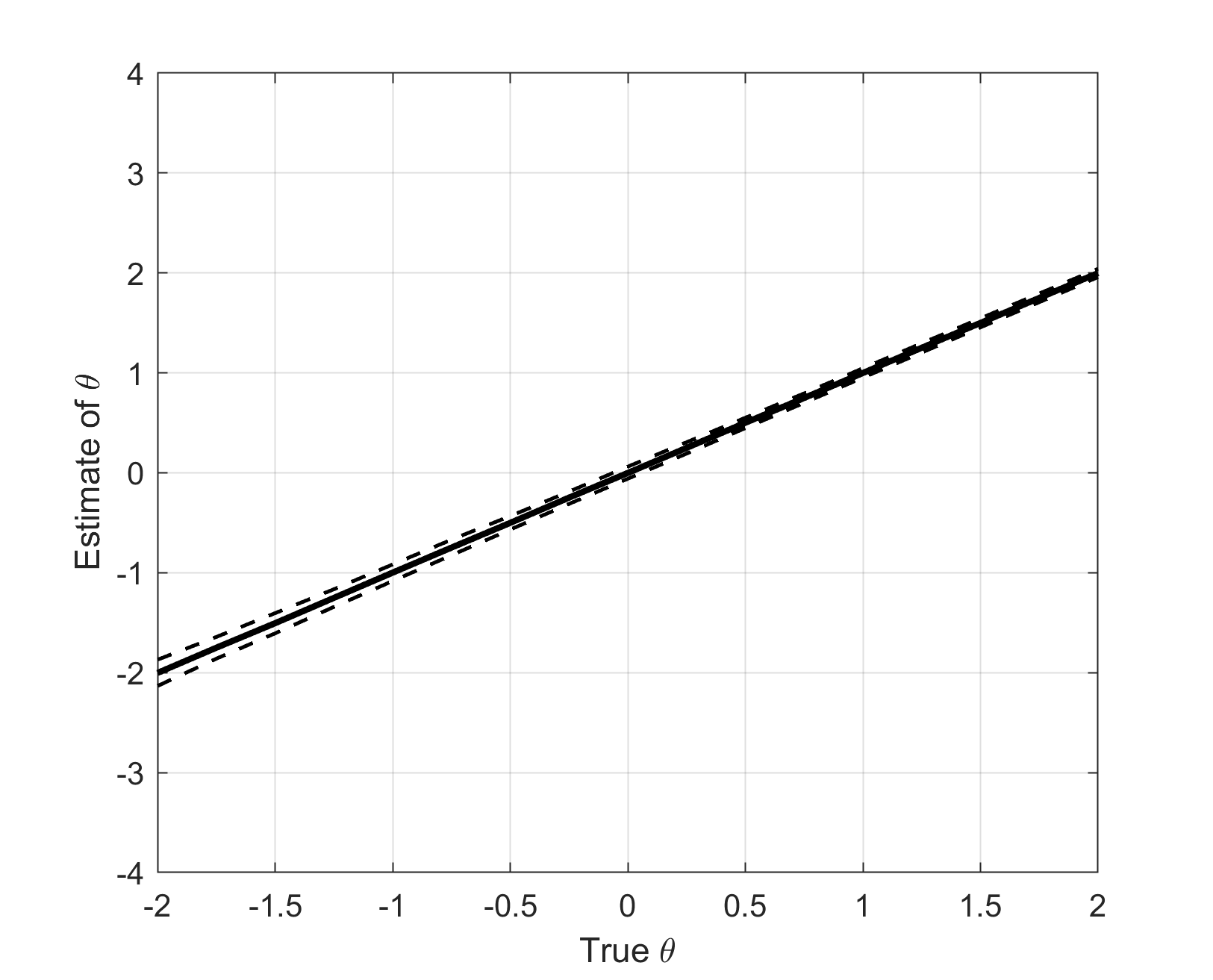}
    \end{tabular}\\
	\caption{\small Plug-in and Neyman-orthogonal estimators in a Poisson model \label{fig_mc_poisson}}
	\floatfoot{\small Notes: Monte Carlo mean (solid), 5th and 95th percentiles (dashed), and the 45-degree line (dotted), as functions of $\theta_0$. $N=1000$, $T=5$, $T_{\rm split}=2$ periods for sample splitting (without cross-fitting), and $1000$ replications. Units with no positive split-sample count are dropped. The Neyman-orthogonal estimator does not require first-step nuisance estimates, but here we still drop units with no positive split-sample count.}
\end{figure}

\paragraph{Example \ref{example_1} (continued).}

In the panel logit model (\ref{eq_logit}), the scores for the slope coefficients and fixed effect have the same structure as in the Poisson model,
\begin{equation*}
	w_1 = \nabla_{\eta}\ln f=\sum_{t=1}^{T}(y_t-\pi_t),\qquad u=\nabla_{\theta}\ln f=\sum_{t=1}^{T}x_t\,(y_t-\pi_t),
\end{equation*}
where $\pi_t=\exp(x_t^{\top}\theta+\eta)/(1+\exp(x_t^{\top}\theta+\eta))$. Let $v_t=\pi_t(1-\pi_t)$ be the conditional variance of $y_t$. Then projecting $u$ orthogonally to $w_1$ delivers
\begin{equation}
	u^*_1=\sum_{t=1}^{T}\left(x_t-\overline{x}\right)(y_t-\pi_t),\qquad \overline{x}=\frac{\sum_{t=1}^T v_t\,x_t}{\sum_{t=1}^T v_t}.
	\label{eq_logit_orth}
\end{equation}
The key difference with the previous example is that, now, the weights $v_t$ in $\overline{x}$ still depend on $\eta$, while the weights $p_t$ in the Poisson model did not.

For fixed $\theta$, $\eta$, and $x$, let
\begin{equation}
	G_{x;\theta,\eta}(\eta')=\mathbb{E}_{x;\theta,\eta}\left[u^*_1(y,x,\theta,\eta')\right]=\sum_{t=1}^{T}\left(x_t-\overline{x}(\eta')\right)\left(\pi_t(\eta)-\pi_t(\eta')\right),
	\label{eq_logit_G}
\end{equation}
where we have made the dependence of $\pi_t$ on $\eta$ explicit.
We have $G_{x;\theta,\eta}(\eta)=0$. Furthermore, by Neyman orthogonality, it equally holds that
$\partial_{\eta'}G_{x;\theta,\eta}(\eta')|_{\eta'=\eta}=0$. However, we have (see Appendix \ref{App_binary_G} for a proof)
\begin{equation}
	\partial_{\eta'}^{2}G_{x;\theta,\eta}(\eta')\big|_{\eta'=\eta}
	=\sum_{t=1}^{T}v_t\left(1-2\pi_t\right)\left(x_t-\overline{x}\right),
	\label{eq_logit_G2}
\end{equation}
where all quantities are evaluated at $\eta$. This quantity differs from zero in general. For example, for $T=2$ and a scalar covariate we have
\begin{equation}
	\partial_{\eta'}^{2}G_{x;\theta,\eta}(\eta')\big|_{\eta'=\eta}
	=\frac{2\left(x_1-x_2\right)v_1v_2\left(\pi_2-\pi_1\right)}{v_1+v_2},\label{eq_logit_G2_T2}
\end{equation}
which is non-zero unless the covariate is constant over time or $\theta=0$. The mean of $u_1^*$ is thus still subject to quadratic noise in $\eta'-\eta$. In this example the nuisance parameter estimate converges at the $T^{-\nicefrac{1}{2}}$ rate, which then translates into a bias of order $T^{-1}$. This is the same order as the bias of the standard maximum likelihood estimator \citep{Arellano2003}.

The panel logit example reveals that we need $N= o(T)$ for Neyman-orthogonal moments to alleviate the incidental-parameter bias in panel data. This rate corresponds to panels that are substantially longer than they are wide.  Furthermore, it is also the same rate required for maximum likelihood to be asymptotically unbiased (see, e.g., \citealt{HahnNewey2004}). As such, first-order orthogonalization of the likelihood score does not improve the order of the bias over estimation based on the score itself. 

It is worth noting that the projection underlying (\ref{eq_logit_orth}) can be described in terms of the sufficient statistic. Indeed,
$$
w_1=\sum_{t=1}^{T}\left(y_t-\pi_t\right)=s-\mathbb{E}_{x;\theta,\eta}\left[s\right],\qquad s=\sum_{t=1}^{T}y_t,
$$
where $s$ is sufficient for $\eta$ at a given $\theta$. The nuisance score and the centered sufficient statistic therefore span the same space, and $u_1^*$ is the residual from projecting $u$ orthogonally to $s$. The same holds in the general logit model (\ref{eq_plrlogit}), where $s=x_2^{\top}y$ and
$$
w_1=\nabla_{\eta}\ln f=x_2^{\top}\left(y-\pi\right)=s-\mathbb{E}_{x;\theta,\eta}\left[s\right].
$$
The same property also holds in the Poisson model of Example \ref{example_2}, where $w_1=\sum_{t=1}^{T}(y_t-\mu_t)=s-\mathbb{E}_{x;\theta,\eta}[s]$ for the same sufficient statistic $s=\sum_{t=1}^{T}y_t$, with the difference that, as (\ref{eq_u1star_0}) shows, in that model Neyman's projection yields full orthogonality.

Figure \ref{fig_mc_logit} illustrates this in a logit model. Details on the design and estimation, including sample splitting, can be again found in Appendix \ref{app_simu}. The left graph shows that the plug-in estimator is biased, and the bias increases with $|\theta_0|$. Moreover, the right graph shows that the Neyman-orthogonal estimator remains biased, although the bias is in the other direction (in this case, towards zero) compared to plug-in. This confirms that the Neyman-orthogonal estimator is not fully orthogonal in the logit model, and shows that the bias can be large in practice.

\begin{figure}[h!]
	\centering
    \begin{tabular}{cc}
    Plug-in & Neyman\\
	\includegraphics[width=.5\textwidth]{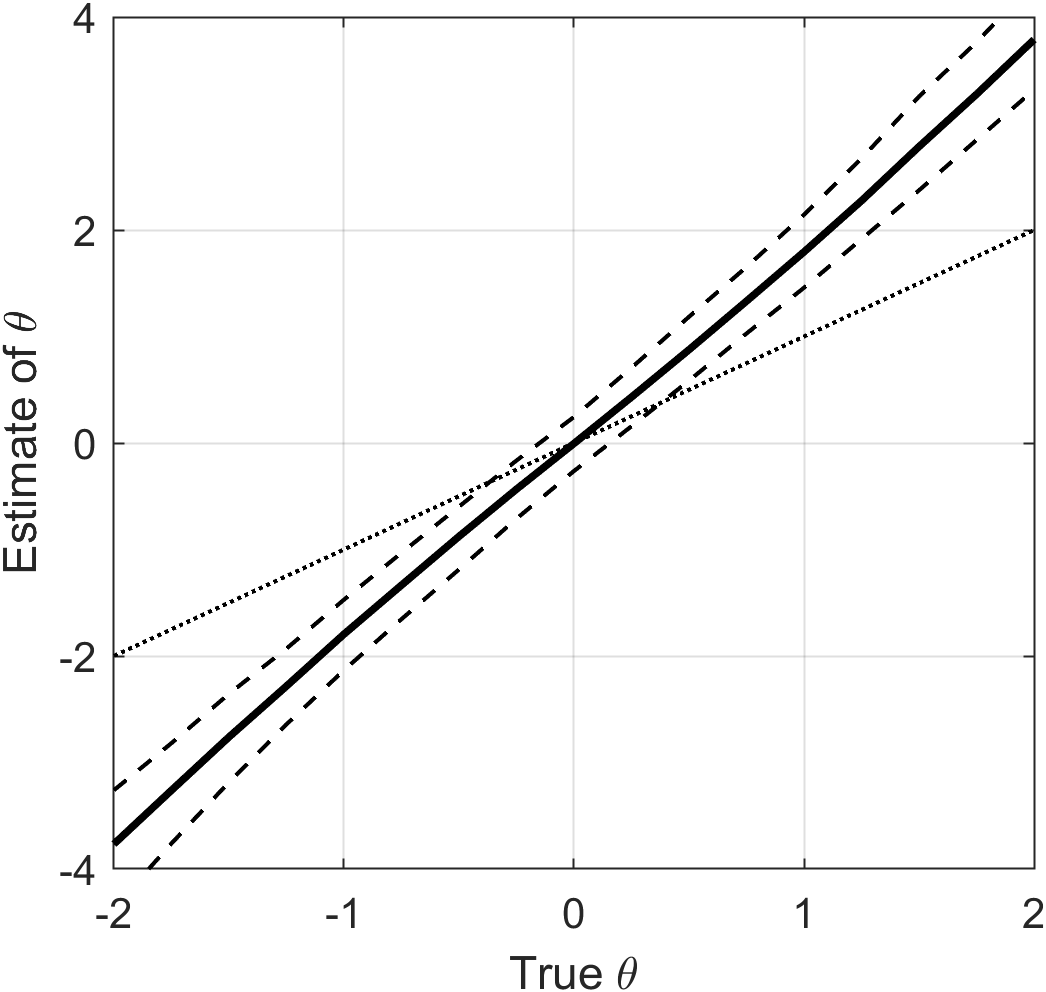}&\includegraphics[width=.5\textwidth]{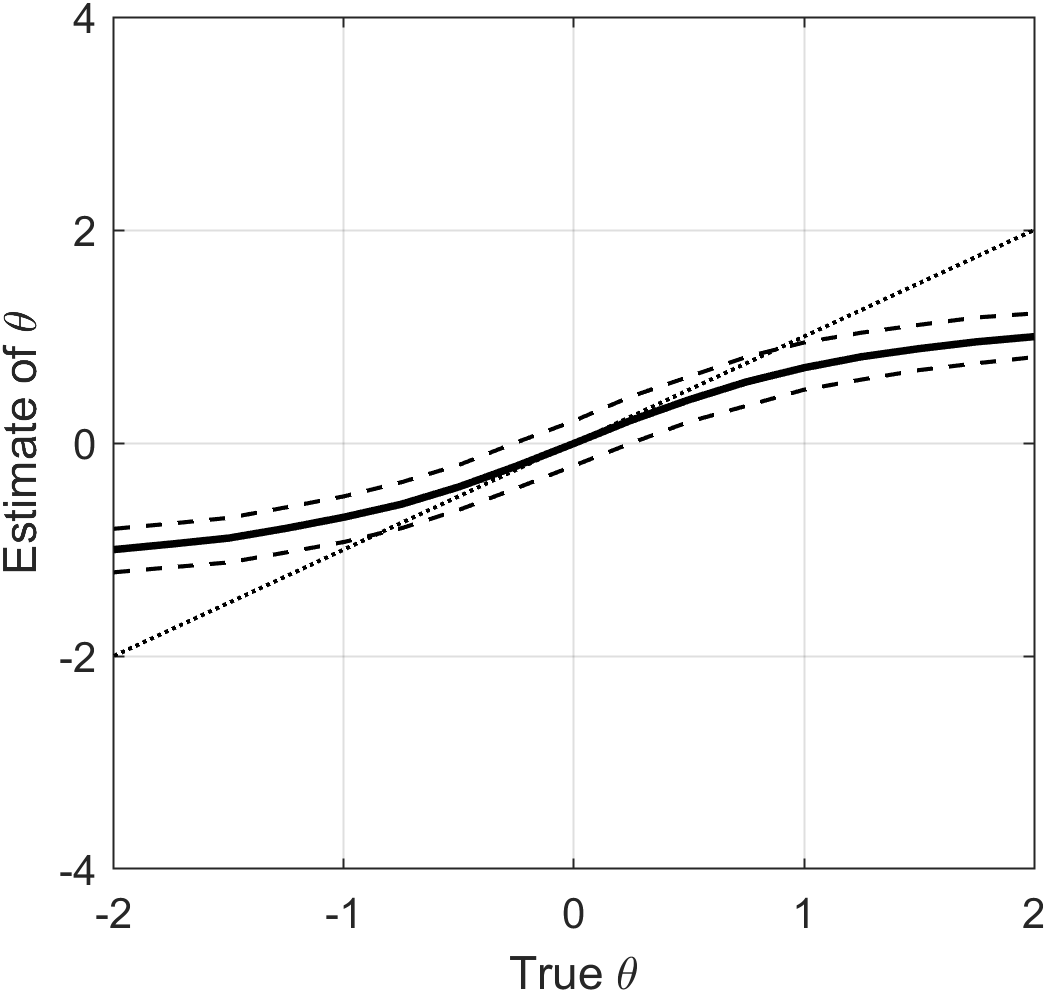}
    \end{tabular}\\
	\caption{\small Plug-in and Neyman-orthogonal estimators in a logit model \label{fig_mc_logit}}
	\floatfoot{\small Notes: see notes to Figure \ref{fig_mc_poisson}. Units whose split-sample outcomes are all zero or all one are dropped.}
\end{figure}

\section{Fully-orthogonal moments\label{sec_full}}
\setcounter{equation}{0}

\subsection{Definition and motivation\label{sec_full_def}}

In the likelihood setting, we require full orthogonality to hold conditionally on $x$ and at all values of the parameters. We thus define a moment function $u^*_\infty$ to be fully orthogonal if 
\begin{equation}\label{eq_full_def}
\mathbb{E}_{x;\theta,\eta}[u^*_\infty(y,x,\theta,\eta')]=0\qquad \text{for all }\theta\in\Theta,\eta\in{\cal{H}},\eta'\in{\cal{H}},\ \text{almost every }x.
\end{equation}
When (\ref{eq_full_def}) holds, the mean of $u^*_\infty$ is completely insensitive to the value of the nuisance parameter at which it is evaluated. Note that (\ref{eq_full_def}) implies $\mathbb{E}_{x;\theta,\eta}[u^*_\infty(y,x,\theta,\eta)]=0$, so a separate unbiasedness requirement is not needed.  

Because (\ref{eq_full_def}) holds at every $\eta'$, one may fix $\eta'$ at any convenient value. We then write $u_{\infty}^*(y,x,\theta)$ for the resulting function, so that full orthogonality reads
\begin{equation}
	\mathbb{E}_{x;\theta,\eta}\left[u_{\infty}^*(y,x,\theta)\right]=0 \qquad\text{for all }\theta\in\Theta,\ \eta\in{\cal{H}},\ \text{almost every }x.
	\label{eq_full_def2}
\end{equation}
In some models a fully-orthogonal moment can be chosen so as not to depend on the value at which $\eta'$ was fixed, as with the panel data moments in differences discussed in the introduction. In general, the projection construction of the next subsection requires setting $\eta'$ to a particular value.

The attractiveness of a fully-orthogonal estimating equation is immediate. Unlike in the Neyman-orthogonal case, GMM estimation based on a fully-orthogonal moment does not require a first-step estimator of $\eta_0$ converging at a specific rate: by (\ref{eq_full_def}) the moment condition holds at every value of the nuisance parameter. Standard GMM theory can then be applied to establish the asymptotic properties of the resulting estimator. In panel data, this leads to estimators that are consistent for fixed $T$ as $N$ tends to infinity, in contrast with estimators based on Neyman orthogonality alone.

\subsection{Orthogonality construction\label{sec_fd}}

Condition \eqref{eq_full_def2} is
\begin{equation}
	\int u^*_\infty(y,x,\theta)\,f(y \,|\, x;\theta,\eta)\,dy=0 \qquad \text{for all }\theta\in\Theta,\eta\in\mathcal{H},\ \text{almost every }x.
	\label{eq_full_lik}
\end{equation}
We will now show that (\ref{eq_full_lik}) is equivalent to $u_{\infty}^*$ being uncorrelated with the closed span of likelihood ratios of the model. This characterization, which is reminiscent of (\ref{eq_neyman_score}) in the Neyman case, will suggest a construction of fully-orthogonal moments based on a projection argument.

To proceed, fix some covariate value $x$ and parameter values $\theta\in\Theta$ and $\eta\in\mathcal{H}$. We denote the likelihood ratios of the model as
\begin{equation*}
	r_{\eta'}(y,x,\theta,\eta)=\frac{f(y\,|\, x;\theta,\eta')}{f(y\,|\, x;\theta,\eta)},\qquad \eta' \in\mathcal{H}.
\end{equation*}
Throughout this section we assume that, for every $x$ and $\theta\in\Theta$, the distributions $f(\cdot\,|\,  x;\theta,\eta)$, $\eta\in\mathcal{H}$, are mutually absolutely continuous, so that the likelihood ratios $r_{\eta'}$ are well defined and the support of $y$ does not vary with $\eta$. For the projection arguments below we further assume that $r_{\eta'}(\cdot,x,\theta,\eta)\in L^2\left(f(\cdot\,|\, x;\theta,\eta)\right)$ for all $\eta,\eta'\in\mathcal{H}$.

With this definition, we have
\begin{align*}
(\ref{eq_full_lik})& \Leftrightarrow \int u^*_\infty(y,x,\theta)\,f(y \,|\, x;\theta,\eta')\,dy=0 \qquad \text{for all }\eta'\in\mathcal{H}\\
& \Leftrightarrow \int u^*_\infty(y,x,\theta)\,\frac{f(y \,|\, x;\theta,\eta')}{f(y \,|\, x;\theta,\eta)}f(y \,|\, x;\theta,\eta)\,dy=0 \qquad \text{for all }\eta'\in\mathcal{H}\\
& \Leftrightarrow \int u^*_\infty(y,x,\theta)\,r_{\eta'}(y,x,\theta,\eta)f(y \,|\, x;\theta,\eta)\,dy=0 \qquad \text{for all }\eta'\in\mathcal{H}.
\end{align*}
This shows that $u^*_\infty$ is fully orthogonal in the sense of (\ref{eq_full_def2}) if and only if
\begin{equation}
	\mathbb{E}_{x;\theta,\eta}\left[u^*_{\infty}(y,x,\theta)\,r_{\eta'}(y,x,\theta,\eta)\right]=0\qquad \text{for all }\theta\in\Theta,\ \eta\in{\cal{H}},\ \eta'\in{\cal{H}},\ \text{almost every }x.
	\label{eq_full_score}
\end{equation}

Condition (\ref{eq_full_score}) is the full-orthogonality counterpart to the Neyman characterization (\ref{eq_neyman_score}). It shows that $u^*_{\infty}$ is fully orthogonal if and only if it is uncorrelated with all likelihood ratios $r_{\eta'}$ for $\eta'\in{\cal{H}}$. Hence, collecting the likelihood ratios into the closed linear span (for given $x$, $\theta$, $\eta$),
\begin{equation}
	{\cal S}_{x;\theta,\eta}=\overline{\operatorname{span}}\left\{\,y\mapsto r_{\eta'}(y,x,\theta,\eta)\ :\ \eta' \in\mathcal{H}\,\right\}\ \subseteq\ L^2\big(f(\cdot\,|\, x;\theta,\eta)\big),
	\label{eq_fd_space}
\end{equation}
it follows from (\ref{eq_full_score}) and continuity of the inner product that $u^*_{\infty}$ is fully orthogonal if and only if it is orthogonal in $L^2$ to the linear space ${\cal S}_{x;\theta,\eta}$.

This characterization suggests the following construction, based on a projection argument. Given any moment function $u$, the projection residual, which now carries the reference value $\eta$ as an argument,
		\begin{equation}
			u^{*}_\infty(\cdot,x,\theta,\eta)=u(\cdot,x,\theta,\eta)-\operatorname{Proj}_{x;\theta,\eta}\left(u(\cdot,x,\theta,\eta)\,\middle\|\,{\cal S}_{x;\theta,\eta}\right),
			\label{eq_fd_proj}
		\end{equation}
        where $\operatorname{Proj}_{x;\theta,\eta}\left[v\,\middle\|\,{\cal S}\right]$ is given by (\ref{eq_proj}), 
		is fully orthogonal and, among all fully-orthogonal functions, minimizes the $L^2(f(\cdot\mid x;\theta,\eta))$ distance to $u$. 

        Note that (\ref{eq_fd_proj}) depends on the value $\eta$, through the inner product used in the projection and, when $u$ depends on it, through the moment function itself. This dependence is of a different nature from the one in the Neyman construction. There, $\eta$ is the unknown nuisance parameter and must be replaced by an estimate. Here, $\eta$ is merely a reference value that we are free to choose: by (\ref{eq_full_lik}), full orthogonality of a given function involves no reference value, so every choice of $\eta$ produces a valid element of the same set. The choice does affect which element is selected, hence the precision of the resulting estimator, but is immaterial for the full orthogonality property.

Unlike the Neyman case where the projection is finite-dimensional (since we focus on a finite-dimensional $\eta$), the projection in (\ref{eq_fd_proj}) is infinite-dimensional in general. The moment function $u^{*}_\infty$ can be expressed as the solution to an operator equation. We provide details in Appendix~\ref{App_variational}.

\subsection{Examples\label{sec_full_ex}}

\paragraph{Example \ref{example_2} (continued).} 
 In the panel Poisson model (\ref{eq_poisson}), a fully-orthogonal moment based on $u=\nabla_{\theta}\ln f$ is $u_{\infty}^*=u-\mathbb{E}_{x;\theta,\eta}[u\,|\, s]$, where $s=\sum_{t=1}^Ty_t$. A key property of this model is that the conditional expectation $\mathbb{E}_{x;\theta,\eta}[u\,|\, s]$ is linear. Indeed, conditionally on $s$ the vector $y$ is multinomial with probabilities $p_t$, so that $\mathbb{E}_{x;\theta,\eta}[u\,|\,s]=\overline{x}\,(s-\sum_{t}\mu_t)$, which lies in the span of $w_1$. Hence, $u_{\infty}^*$ coincides with the Neyman-orthogonal score $u_1^*$ in (\ref{eq_poisson_orth}).

\paragraph{Example \ref{example_1} (continued).}
Consider the general logistic regression model (\ref{eq_plrlogit}), for a vector of outcomes $y\in\{0,1\}^n$. The likelihood is
\begin{equation*}
	f(y\mid x;\theta,\eta)=\prod_{i=1}^{n}\pi_i^{y_i}(1-\pi_i)^{1-y_i},\qquad \pi_i=\frac{\exp(x_{1i}^{\top}\theta+x_{2i}^{\top}\eta)}{1+\exp(x_{1i}^{\top}\theta+x_{2i}^{\top}\eta)},
\end{equation*}
the rows of the matrices $x_1$ and $x_2$ being $x_{1i}^{\top}$ and $x_{2i}^{\top}$ for $i=1,\ldots,n$. In this exponential family in $\eta$ at fixed $\theta$, with sufficient statistic
$s=x_2^{\top}y\ \in\mathbb{R}^{d_\eta}$, the likelihood ratios are given by
\begin{equation}
	r_{\eta'}(y,x,\theta,\eta)=\frac{f(y\mid x;\theta,\eta')}{f(y\mid x;\theta,\eta)}=c(x,\theta,\eta',\eta)\,\exp\left(s^{\top}(\eta'-\eta)\right),
	\label{eq_logit_ratio}
\end{equation}
where
$$
c(x,\theta,\eta',\eta)=\prod_{i=1}^{n}\frac{1+\exp(x_{1i}^{\top}\theta+x_{2i}^{\top}\eta)}{1+\exp(x_{1i}^{\top}\theta+x_{2i}^{\top}\eta')}
$$
does not depend on $y$. The likelihood ratios are thus, up to a normalization, exponentials in the sufficient statistic $s$, and
\begin{equation}
	{\cal S}_{x;\theta,\eta}=\overline{\operatorname{span}}\left\{\,y\mapsto\exp\left(s^{\top}\tau\right)\ :\ \tau\in\mathbb{R}^{d_\eta}\,\right\}.
	\label{eq_logit_span}
\end{equation}
Since $s$ has finite support this span is the set of all functions of $s$, so that
\begin{equation}
	u^*_\infty(y,x,\theta,\eta)=u(y,x,\theta,\eta)-\mathbb{E}_{x;\theta,\eta}\left[u(y,x,\theta,\eta)\mid s\right].
	\label{eq_logit_full}
\end{equation}

Hence, the logit model provides an example where fully-orthogonal moments (in fact, all of them) can be obtained by de-meaning. However, this does not mean that the resulting moments are automatically informative about $\theta$. Indeed, nontrivial moments are available if and only if $y\mapsto x_2^{\top}y$ is a non-injective mapping; that is, if two distinct outcome configurations share the same value of the sufficient statistic. Suppose that $x_2$ has continuously-distributed entries. Then, the $2^n$ values $\{x_2^{\top}y\,:\, y\in\{0,1\}^n\}$ are almost surely  distinct, so the map $y\mapsto x_2^{\top}y$ is injective, and (\ref{eq_logit_full}) always returns $u_{\infty}^*=0$ irrespective of the original moment $u$. It follows that a nontrivial fully-orthogonal moment does not exist.

Consider now the panel logit model (\ref{eq_logit}). In this case the nuisance parameter is a set of fixed effects, and so the $x_2$ matrix in the previous example is a matrix of unit indicators. Here, $T\geq 2$ already suffices to obtain nontrivial fully-orthogonal moments provided $x_{it}$ varies within units. Indeed, $\binom{T}{c}$ distinct outcome sequences all yield $\sum_t y_t=c$. The conditional logit score is one such moment.  

Consider next the $\beta$-model (\ref{eq_netlogit}) of network formation. In this model, the degree sequence $d_i=\sum_{j\neq i}y_{ij}$, $i=1,...,N$, serves as sufficient statistic for $\eta$ at given $\theta$. The non-injectivity of the degree map can be shown by the following invariance argument, which underlies the estimators of \citet{graham2017econometric} and \citet{Jochmans2018}: for any four distinct agents $\{i,j,k,l\}$, the configurations
\begin{equation*}
	\{y_{ij}=y_{kl}=1,\ y_{ik}=y_{jl}=0\}\qquad\text{and}\qquad \{y_{ij}=y_{kl}=0,\ y_{ik}=y_{jl}=1\}
\end{equation*}
leave the degree sequence unchanged. This delivers informative fully-orthogonal moments based on network tetrads.

\subsection{A mixture interpretation\label{sec_functional}}

The projection in (\ref{eq_fd_proj}) that constructs fully-orthogonal moments can be understood as implementing the construction in \citet{Neyman1959}, applied to a different nuisance parameter and with a different reference measure. In turn, this approach is closely related to functional differencing, which was initially proposed for panel data models in \citet{bonhomme2012functional}.

To describe this interpretation, let us treat $\eta$ as a random vector, drawn from an unrestricted conditional distribution given covariates, $\rho(\eta\,|\, x)$. This perspective is natural in panel data, where it is sometimes referred to as ``correlated random effects'', but we do not need to restrict ourselves to the panel data setting. The model then becomes a semiparametric mixture, with likelihood
\begin{equation}
	f^{\rm mixt}(y\mid x;\theta,\rho)=\int_{\cal{H}} f(y\mid x;\theta,\eta)\,\rho(\eta\mid x)\,d\eta.
	\label{eq_mixture}
\end{equation}
We leave $\rho$ unrestricted, aside from the requirement that ${\cal H}$ be its support. Note that $f^{\rm mixt}$ depends linearly on $\rho$. 

In this mixture model, $u^*$ being fully orthogonal as in (\ref{eq_full_lik}) can be interpreted as $u^*(y,x,\theta)$ having zero mean conditional on $x$ \emph{and} $\eta$, almost surely in $x,\eta$. Hence, $u^*$ being fully orthogonal is equivalent to it belonging to the null space of the conditional expectation operator
$$E_{x;\theta}:u(y)\mapsto E_{x;\theta}[u](\eta)=\int u(y)f(y\,|\, x;\theta,\eta)dy.$$
Finding such functions is the goal of functional differencing. Elements in the null-space of $E_{x;\theta}$ can be found analytically in certain models, as shown by \citet{honore2025moment} and \citet{dano2023transition} in a variety of discrete choice panel models and in \citet{bonhomme2024functional} in some network models. However, as we have seen in Example \ref{example_1}, fully-orthogonal moments may be uninformative about the parameter of interest. \citet{chamberlain2010binary} highlights that most binary choice models only possess trivial fully-orthogonal moments outside the logit case.

We now show that the construction in \citet{Neyman1959}, in the mixture model (\ref{eq_mixture}), provides a general way to construct such functions, which coincides with the functional differencing projection in \citet{bonhomme2012functional}. To see this, note that in the mixture model the nuisance parameter is $\rho$, and the score with respect to it is the collection
\begin{equation}
	\nabla_{\delta}\ln f^{\rm mixt}(y\mid x;\theta,\rho_{\delta})\Big|_{\delta=0}=\frac{\int f(y\mid x;\theta,\eta)\,\nabla_{\delta}\ln\rho_\delta(\eta\mid x)|_{\delta=0}\,\rho(\eta\mid x)\,d\eta}{\int f(y\mid x;\theta,\eta)\,\rho(\eta\mid x)\,d\eta},\label{eq_scores_mixt}
\end{equation}
for all sufficiently regular one-dimensional submodels $\{\rho_\delta\}$ passing through $\rho$ (see Chapter 25 in \citealp{van2000asymptotic}). Neyman's orthogonal moment, see (\ref{eq_neyman_proj}), is then the projection of $u$ onto the orthogonal complement to the span of these scores, the \emph{orthocomplement of the tangent set} of the semiparametric model.

Now, since the quantities $\nabla_{\delta}\ln\rho_\delta(\eta\mid x)|_{\delta=0}$ in (\ref{eq_scores_mixt}) are essentially unrestricted, the zero-mean functions in the orthocomplement of the tangent set at given $x$ and $\theta$ are
$$\left\{u^*\in L^2\left(f^{\rm mixt}(\cdot\mid x;\theta,\rho)\right)\,:\, \int u^*(y)\,f(y \,|\, x;\theta,\eta)\,dy=0 \qquad \text{for all }\eta\in\mathcal{H}\right\},$$
that is, they are all the fully-orthogonal moments that are square integrable with respect to $f^{\rm mixt}(\cdot\mid x;\theta,\rho)$.\footnote{To see this, note that $u^*$ is orthogonal to the score (\ref{eq_scores_mixt}) if and only if, by Fubini's theorem,
$$
\int_{{\cal H}}\nabla_{\delta}\ln\rho_\delta(\eta\mid x)\big|_{\delta=0}\left(\int u^*(y)\,f(y\mid x;\theta,\eta)\,dy\right)\rho(\eta\mid x)\,d\eta=0.
$$
As the submodel scores $\nabla_{\delta}\ln\rho_\delta(\cdot\mid x)|_{\delta=0}$ range over the mean-zero functions of $\eta$, this holds for all submodels if and only if $\int u^*(y)f(y\mid x;\theta,\eta)dy$ does not vary with $\eta$ on the support ${\cal{H}}$ of $\rho$. Requiring in addition that $u^*$ have mean zero under $f^{\rm mixt}$ sets the constant to zero.} This shows that Neyman's strategy applied to the semiparametric mixture model, and the fully-orthogonal projection in (\ref{eq_fd_proj}), coincide, except for the fact that the two projections use different reference measures -- $f^{\rm mixt}(\cdot\mid x;\theta,\rho)$ for the former and $f(\cdot\mid x;\theta,\eta)$ for the latter. Furthermore, the projection onto the orthocomplement of the tangent set of the mixture model coincides with the functional differencing projection. \citet{bonhomme2012functional} proposes a numerical strategy to approximate this projection for estimation. 

\begin{remark}[Efficiency] The tangent set interpretation reveals that, when $u=\nabla_{\theta}\ln f^{\rm mixt}$, the corresponding projection residual is the efficient score of the semiparametric mixture model. While achieving efficiency requires knowledge of the mixing distribution $\rho$, the characterization remains useful in practice. One may compute the projection at a convenient density $\widetilde\rho$ rather than at the true $\rho_0$: the resulting moment is still fully orthogonal, since full orthogonality does not depend on the reference measure, and it is efficient under $\widetilde\rho$. The choice of $\widetilde\rho$ thus plays the role of a working model, guiding the selection among fully-orthogonal moments without affecting validity.
\end{remark}

\begin{remark}[Dynamics] The mixture formulation extends to richer conditioning structures. \citet{bonhomme2023identification,bonhomme2025moment} apply it to nonlinear panel models in which strict exogeneity of covariates is relaxed to allow for dynamic feedback. There, the nuisance parameter is not only the density of the heterogeneity (given initial conditions), but also the density of covariates given past covariates, past outcomes, and the heterogeneity, also known as the feedback process.
\end{remark}

\section{Higher-order orthogonality\label{sec_higher}}
\setcounter{equation}{0}

We now consider a third notion of orthogonality that is intermediate between the previous two.

\subsection{Motivation\label{sec_higher_def}}

To motivate \eqref{eq_neymanq}, consider first a moment function $u_2^*$ orthogonal to order $q=2$. Suppose in this subsection that the nuisance parameter and the moment function are both univariate, so $d_{\eta}=d_u=1$. Expanding $u_2^*$ around $\eta_0$ (this time to second order) and taking expectations while noting that $u_2^*$ has mean zero at the truth, gives
\begin{equation*}
\begin{split}
\mathbb{E}[u_2^*(y,x,\theta_0,\widehat\eta)]
 = \hphantom{\frac{1}{2}} &
\mathbb{E}[\nabla_\eta u_2^*(y,x,\theta_0,\eta_0) \, (\widehat\eta - \eta_0)^{\hphantom{2}}]
\\
& +
\frac{1}{2} 
\mathbb{E}[\nabla_\eta^2 u_2^*(y,x,\theta_0,\eta_0) \, (\widehat\eta - \eta_0)^2]
+
(\text{higher-order terms}) .
\end{split}
\end{equation*}

When $\widehat{\eta}$ is independent of the data (for example, estimated on a separate sample, as explained in Subsection \ref{sec_rationale}), this gives
\begin{equation}
\begin{split}
\mathbb{E}[u_2^*(y,x,\theta_0,\widehat\eta)]
 = \hphantom{\frac{1}{2}} &
\underset{\text{=0 under \eqref{eq_neymanq} at $j=1$}}{\underbrace{\mathbb{E}[\nabla_\eta u_2^*(y,x,\theta_0,\eta_0)]}} \, \mathbb{E}[(\widehat\eta - \eta_0)^{\hphantom{2}}]
\\
& +
\frac{1}{2} 
\underset{\text{=0 under \eqref{eq_neymanq} at $j=2$}}{\underbrace{\mathbb{E}[\nabla_\eta^2 u_2^*(y,x,\theta_0,\eta_0) ]}}\, \mathbb{E}[(\widehat\eta - \eta_0)^2]
+
(\text{higher-order terms}) .
\end{split}\label{eq_second_order}
\end{equation}

Hence, since $u_2^*$ is orthogonal to order $q=2$, the two dominant terms in (\ref{eq_second_order}) are equal to zero. Furthermore, the remainder term will generally be of order $\mathbb{E}[\|\widehat\eta - \eta_0\|^3]$, implying that the bias induced by estimation noise in the nuisance parameter is asymptotically negligible when $\sqrt{n} \, \mathbb{E}[\|\widehat\eta - \eta_0\|^3] \rightarrow 0$ as $n\rightarrow \infty$.

More generally, while the case $q=1$ recovers the classical Neyman-orthogonality condition, larger values of $q$ yield increased robustness to estimation noise in the nuisance parameter. By the same logic, if $u_q^*$ is orthogonal to order $q$, then $\mathbb{E}[u_q^*(y,x,\theta_0,\widehat\eta)]$ is of order $\mathbb{E}[\|\widehat\eta - \eta_0\|^{q+1}]$, implying that the bias induced by estimation noise in the nuisance parameter is asymptotically negligible when $\sqrt{n} \, \mathbb{E}[\|\widehat\eta - \eta_0\|^{q+1}] \rightarrow 0$ as $n\rightarrow \infty$. This requires that the first-step estimator $\widehat\eta$ converge at a rate faster than $n^{-\nicefrac{1}{2(q+1)}}$, thus relaxing Neyman's $n^{-\nicefrac{1}{4}}$ condition when $q>1$ \citep{MackeySyrgkanisZadik2018}.

\subsection{Construction by projection on generalized scores\label{sec_higher_constr}}

In the likelihood setting, we require orthogonality to order $q$ to hold conditionally on $x$ and at all values of the parameters, similarly to our definitions of Neyman and full orthogonality. That is, we call $u_q^*$ orthogonal to order $q$ if it is conditionally unbiased,
\begin{equation}
\mathbb{E}_{x;\theta,\eta}\left[u^*_q(y,x,\theta,\eta)\right]=0 \qquad \text{for all }\theta\in\Theta,\ \eta\in{\cal{H}},\ \text{almost every }x,\label{eq_mom_star_lik_q}
\end{equation}
and all its first $q$ derivatives in $\eta$ have zero expectation as well,
\begin{equation}
\mathbb{E}_{x;\theta,\eta}\left[\nabla_{\eta}^ju^*_q(y,x,\theta,\eta)\right]=0 \qquad \text{for all }j=1,...,q,\ \theta\in\Theta,\ \eta\in{\cal{H}},\ \text{almost every }x,\label{eq_qorth}
\end{equation}
where recall that $\nabla_{\eta}^j$ collects all derivatives at order $j$ in a vector.

Orthogonality to order $q$ can be equivalently written as an uncorrelatedness condition, similarly to the other two notions of orthogonality we have considered. To see this, recall the likelihood ratio, for given $x$ and $\theta$,
\begin{equation*}
	r_{\eta'}(y,x,\theta,\eta)=\frac{f(y\,|\, x;\theta,\eta')}{f(y\,|\, x;\theta,\eta)},\qquad \eta' \in\mathcal{H}.
\end{equation*}
Denote its vector of $j$th derivatives as
$$
w_j (y,x,\theta,\eta) = \nabla_{\eta'}^j\, r_{\eta'}(y,x,\theta,\eta)\big|_{\eta'=\eta}=\frac{\nabla_{\eta}^jf(y\,|\, x;\theta,\eta)}{f(y\,|\, x;\theta,\eta)}.
$$

Observe that, when $j=1$, $w_j$ coincides with the score for the nuisance parameter, $w_1$, from Section \ref{sec_neyman}. When $j=2$, 
$$w_2 (y,x,\theta,\eta) = \nabla^2_{\eta}\ln f(y\,|\, x;\theta,\eta)+\nabla_{\eta}\ln f(y\,|\, x;\theta,\eta)\otimes\nabla_{\eta}\ln f(y\,|\, x;\theta,\eta),$$
which has mean zero by the information identity. The set of functions $\lbrace w_j \rbrace$ is known as the \cite{Bhattacharyya1946} basis.\footnote{More generally, $w_j$ is a polynomial in $\nabla_{\eta}\ln f$,..., $\nabla_{\eta}^j\ln f$. In the scalar case, it is given by the complete Bell polynomial evaluated at these elements. The vector case follows from the multivariate Fa\`a di Bruno formula (\citealp{constantine1996multivariate}).} Its elements can be interpreted as generalized score functions, and they satisfy (by differentiating $j$ times $\int r_{\eta'}f_{\eta}dy=1$ at $\eta'=\eta$):
\begin{equation}
\mathbb{E}_{x;\theta,\eta}[w_j (y,x,\theta,\eta)] = 0 \qquad \text{for all }j\geq 1,\ \theta\in\Theta,\ \eta\in{\cal{H}},\ \text{almost every }x.\label{eq_zeromean_wj}
\end{equation}

Now consider again $q=2$ in the all-scalar case. As in our derivation of Neyman's original device, differentiate the unbiasedness condition (\ref{eq_mom_star_lik_q}) once with respect to $\eta$. This yields the moment condition
\begin{equation}
\underset{\text{=0 under \eqref{eq_qorth} at $j=1$}}{\underbrace{\mathbb{E}_{x;\theta,\eta}\left[\nabla_{\eta}u^*_2(y,x,\theta,\eta)\right]}}+\mathbb{E}_{x;\theta,\eta}\left[u^*_2(y,x,\theta,\eta)\,w_{1}(y,x,\theta,\eta)\right]=0,
	\label{eq_q2_d1}
    \end{equation}
so
\begin{equation} \label{eq:w10}
\mathbb{E}_{x;\theta,\eta}\left[u^*_2(y,x,\theta,\eta)\,w_{1}(y,x,\theta,\eta)\right]=0  \text{ for all }\theta\in \Theta,\ \eta\in {\cal{H}},\ \text{ almost every }x.
\end{equation}

Next, twice differentiate (\ref{eq_mom_star_lik_q}) to see that 
\begin{align*}
&\underset{\text{=0 under \eqref{eq_qorth} at $j=2$}}{\underbrace{\mathbb{E}_{x;\theta,\eta}\left[\nabla_{\eta}^2u^*_2(y,x,\theta,\eta)\right]}}\\
&+
2\,\mathbb{E}_{x;\theta,\eta}\left[\nabla_{\eta}u^*_2(y,x,\theta,\eta)\,w_{1}(y,x,\theta,\eta)\right]+\mathbb{E}_{x;\theta,\eta}\left[u^*_2(y,x,\theta,\eta)\,w_{2}(y,x,\theta,\eta)\right]=0,
\end{align*}
so
\begin{align}
&2\,\mathbb{E}_{x;\theta,\eta}\left[\nabla_{\eta}u^*_2(y,x,\theta,\eta)\,w_{1}(y,x,\theta,\eta)\right]+\mathbb{E}_{x;\theta,\eta}\left[u^*_2(y,x,\theta,\eta)\,w_{2}(y,x,\theta,\eta)\right]=0.\label{eq_inter_q}
\end{align}

Furthermore, because \eqref{eq:w10} can be written as $\int u_2^*\nabla_{\eta}f\, dy=0$ and needs to hold at all parameter values, we obtain, by differentiating it, 
\begin{equation}
	\mathbb{E}_{x;\theta,\eta}\left[\nabla_{\eta}u^*_2(y,x,\theta,\eta)\,w_{1}(y,x,\theta,\eta)\right]+\mathbb{E}_{x;\theta,\eta}\left[u^*_2(y,x,\theta,\eta)\,w_{2}(y,x,\theta,\eta)\right]=0.
	\label{eq_q2_cross}
\end{equation}
Combining (\ref{eq_inter_q}) and (\ref{eq_q2_cross}) then implies
\begin{equation} \label{eq:w20}
\mathbb{E}_{x;\theta,\eta}\left[u^*_2(y,x,\theta,\eta)\,w_{2}(y,x,\theta,\eta)\right]=0 .
\end{equation}

This argument shows that, if $u_2^*$ is orthogonal to order $q=2$, then it is uncorrelated with the first two Bhattacharyya functions $w_1$ and $w_2$. Moreover, by following the same argument in reverse it is easy to see that the converse is also true. Hence, second-order orthogonality, for a moment satisfying (\ref{eq_mom_star_lik_q}), can be equivalently expressed as (\ref{eq:w10}) and (\ref{eq:w20}).

Iterating the argument, now for the general vector-valued case, reveals that $u_q^*$ is orthogonal to order $q$ if and only if it satisfies the unbiasedness condition (\ref{eq_mom_star_lik_q}) and the orthogonality conditions
\begin{equation}
	\mathbb{E}_{x;\theta,\eta} [u^*_q(y,x,\theta,\eta)\, w_j(y,x,\theta,\eta)^{\top} ] =0,\qquad j=1,\ldots,q . 
	\label{eq_qorth_lik}
\end{equation}

The representation in (\ref{eq_qorth_lik}) suggests constructing higher-order orthogonal functions by projection, as we did before for the other two notions of orthogonality. Specifically, $u_q^*$ is constructed as the residual of a projection of the original moment function $u$ onto the span of $w_1,\ldots, w_q$:
\begin{equation}
			u^{*}_q(\cdot,x,\theta,\eta)=u(\cdot,x,\theta,\eta)-\operatorname{Proj}_{x;\theta,\eta}\left(u(\cdot,x,\theta,\eta)\,\middle\|\, \operatorname{span}\{w_1(\cdot,x,\theta,\eta),\ldots,w_q(\cdot,x,\theta,\eta)\}\right),
			\label{eq_q_proj}
		\end{equation}
        where $\operatorname{Proj}_{x;\theta,\eta}\left[v\,\middle\|\,{\cal S}\right]$ is again given by (\ref{eq_proj}).

Stacking $W_q=(w_1^{\top},\ldots,w_q^{\top})^{\top}$, we can write the explicit expression 
\begin{equation}
	u^{*}_{q}(y,x,\theta,\eta)=u(y,x,\theta,\eta)-\Lambda_q(x;\theta,\eta) W_{q}(y,x,\theta,\eta),
	\label{eq_qorth_explicit}
\end{equation}
with least-squares projection coefficient
\begin{equation*}
	\Lambda_q(x;\theta,\eta)=\mathbb{E}_{x;\theta,\eta}\left[u(y,x,\theta,\eta)\, W_q(y,x,\theta,\eta)^{\top}\right] \, \mathbb{E}_{x;\theta,\eta}\left[W_q(y,x,\theta,\eta)\, W_q(y,x,\theta,\eta)^{\top}\right]^{-1}.
\end{equation*}
Note that, by (\ref{eq_zeromean_wj}), $u_q^*$ automatically satisfies (\ref{eq_mom_star_lik_q}) provided the original moment $u$ is unbiased at all parameter values.  

The explicit expression in (\ref{eq_qorth_explicit}) requires nonsingularity of the Gram matrix of the leading basis functions $W_q$, $\mathbb{E}_{x;\theta,\eta}\left[W_q(y,x,\theta,\eta)\, W_q(y,x,\theta,\eta)^{\top}\right]$, as Subsection \ref{sec_higher_limit} below illustrates. For the special case where $u = \nabla_\theta \ln f$, (\ref{eq_qorth_explicit}) reduces to the projected score of \citet{SmallMcLeish1989} and \citet{WatermanLindsay1996}. The more general construction is introduced in \cite{bonhomme2024neyman}.

\begin{remark}[Nested spaces]
Note that, as anticipated in Figure \ref{fig_graph}, the three notions of orthogonality correspond to a sequence of nested spaces: 
$$ \underset{\text{Neyman}}{\underbrace{\operatorname{span}\{w_1(\cdot,x,\theta,\eta)\}}}\subseteq \underset{\text{order }q}{\underbrace{\operatorname{span}\{w_1(\cdot,x,\theta,\eta),\ldots,w_q(\cdot,x,\theta,\eta)\}}}\subseteq \underset{\text{full}}{\underbrace{{\cal S}_{x;\theta,\eta}}},$$
where ${\cal S}_{x;\theta,\eta}$ is the closed span of likelihood ratios defined in (\ref{eq_fd_space}). The second inclusion holds because each $w_j$ is a limit of finite differences of likelihood ratios. A larger space brings more robustness to estimation noise in $\eta$, at the cost of less information about $\theta$.
\end{remark}

\subsection{Examples\label{sec_higher_ex}}

\paragraph{Example \ref{example_1} (continued).}
Consider again the general logistic regression model (\ref{eq_plrlogit}), for a vector of outcomes $y\in\{0,1\}^n$. Writing the log-likelihood as
$$
\ln f(y\mid x;\theta,\eta)=y^{\top}x_1\theta+s^{\top}\eta-\kappa(x,\theta,\eta),$$
for
$$\kappa(x,\theta,\eta)=\sum_{i=1}^{n}\ln\left(1+\exp\left(x_{1i}^{\top}\theta+x_{2i}^{\top}\eta\right)\right),
$$
we have $\nabla_{\eta}\ln f=s-\nabla_{\eta}\kappa$, for $s=x_2^{\top}y$ the sufficient statistic, while $\nabla_{\eta}^j\ln f=-\nabla_{\eta}^j\kappa$ for $j\geq 2$ does not depend on $y$. Hence $w_j$ is a polynomial of degree $j$ in the sufficient statistic $s$, its degree-$j$ term being $w_1^{\otimes j}$.

The coefficients of this polynomial can be expressed in terms of the cumulants of $s$. Indeed, $\kappa$ is the cumulant generating function of the exponential family in $\eta$, so that $\nabla_{\eta}^{j}\kappa$ is the $j$th cumulant of $s$ under $f(\cdot\mid x;\theta,\eta)$. In particular $\nabla_{\eta}\kappa=\mathbb{E}_{x;\theta,\eta}[s]$ and $\nabla^2_{\eta}\kappa=\operatorname{vec}\left(\operatorname{Var}_{x;\theta,\eta}(s)\right)$, so that
\begin{equation*}
	w_1=s-\mathbb{E}_{x;\theta,\eta}\left[s\right],\qquad w_2=\left(s-\mathbb{E}_{x;\theta,\eta}\left[s\right]\right)^{\otimes 2}-\operatorname{vec}\left(\operatorname{Var}_{x;\theta,\eta}(s)\right).
\end{equation*}

Since the $w_j$ have mean zero and are of degree $j$, they form a basis of the mean-zero polynomials in $s$ of total degree at most $q$. Orthogonality to order $q$ therefore amounts to projecting the original moment on powers of the sufficient statistic:
\begin{equation}
	u^*_q(\cdot,x,\theta,\eta)=u(\cdot,x,\theta,\eta)-\operatorname{Proj}_{x;\theta,\eta}\left(u(\cdot,x,\theta,\eta)\,\middle\|\,\left\{\text{polynomials in }s\text{ of degree}\leq q\right\}\right),
	\label{eq_logit_q}
\end{equation}
where the constants may be included in the projection space without affecting the residual, since $u$ has mean zero.

Specializing to the panel logit model (\ref{eq_logit}) makes the construction fully explicit. At the level of a single unit, $x_2$ is a vector of ones and the sufficient statistic is the scalar $s=\sum_{t=1}^{T}y_t$, so that $w_j$ is a polynomial of degree $j$ in $s$ whose coefficients are the cumulants of a sum of $T$ independent Bernoulli variables; Appendix~\ref{App_binary} collects the explicit formulas. Orthogonality to order $q$ therefore projects the score on the powers $s,s^2,\ldots,s^q$.

Figure \ref{fig_mc_logit_2} illustrates this in Monte Carlo simulations based on the logit model, using the design mentioned in Section \ref{sec_neyman}. The left graph shows the estimator orthogonal to $q=2$, the middle one shows the estimator orthogonal to $q=3$, and the right one shows the conditional logit estimator, which is based on a fully-orthogonal moment. We see that the estimators behave very similarly in this simulation. Hence, even if there is no guarantee for the second-order orthogonal moment to be fully orthogonal in the logit model, increasing the orthogonality requirement by a single order (from $q=1$ to $q=2$) is sufficient to almost fully remove the bias of the Neyman-orthogonal estimator.

\begin{figure}[h!]
	\centering
    \begin{tabular}{ccc}
 $q=2$ & $q=3$ & Conditional logit\\
	\includegraphics[width=.33\textwidth]{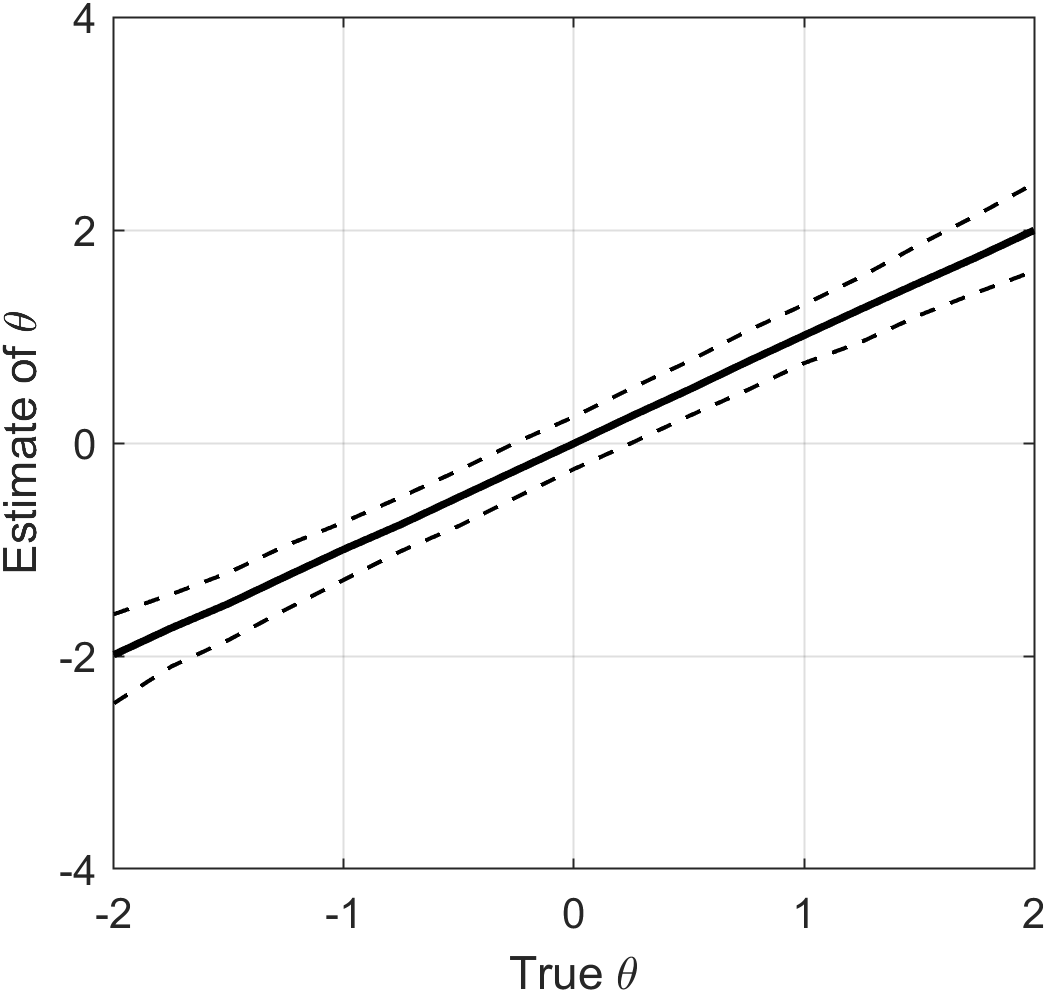}&\includegraphics[width=.33\textwidth]{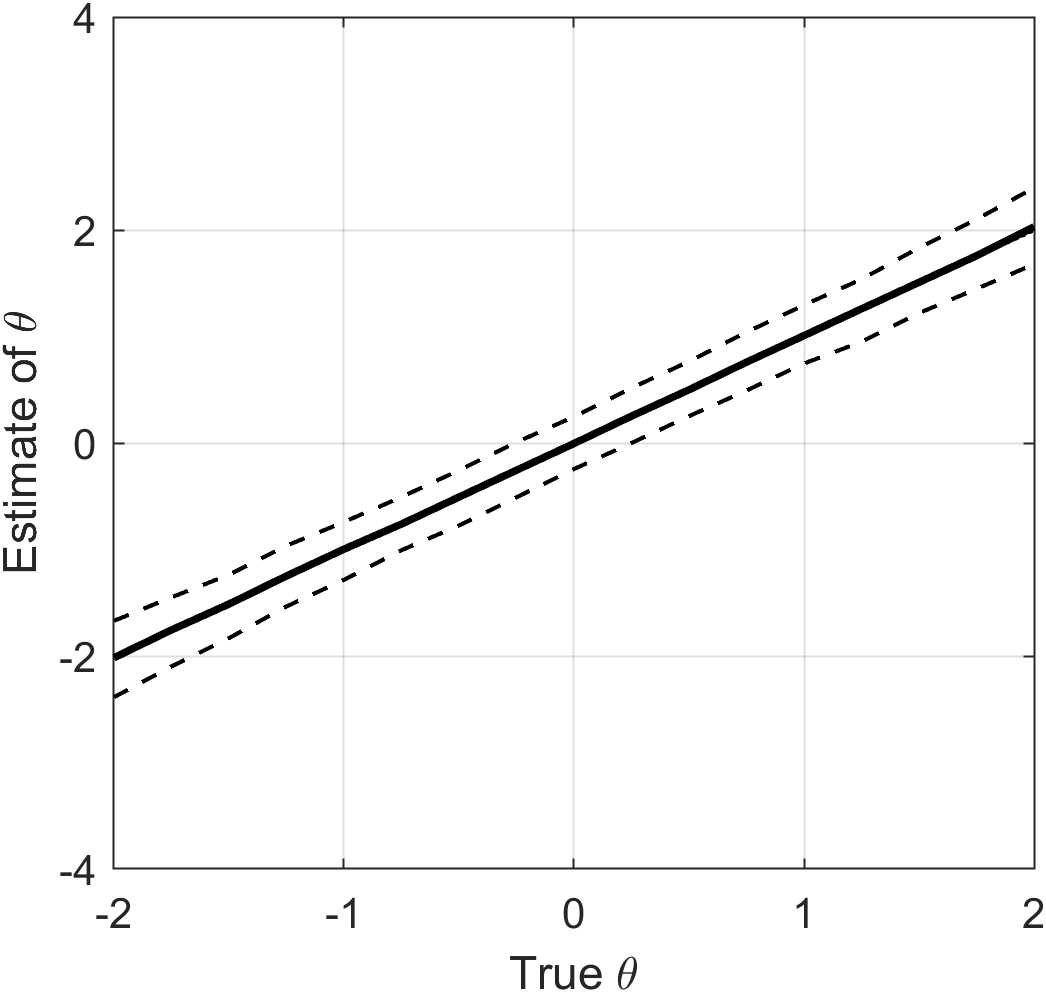}&\includegraphics[width=.33\textwidth]{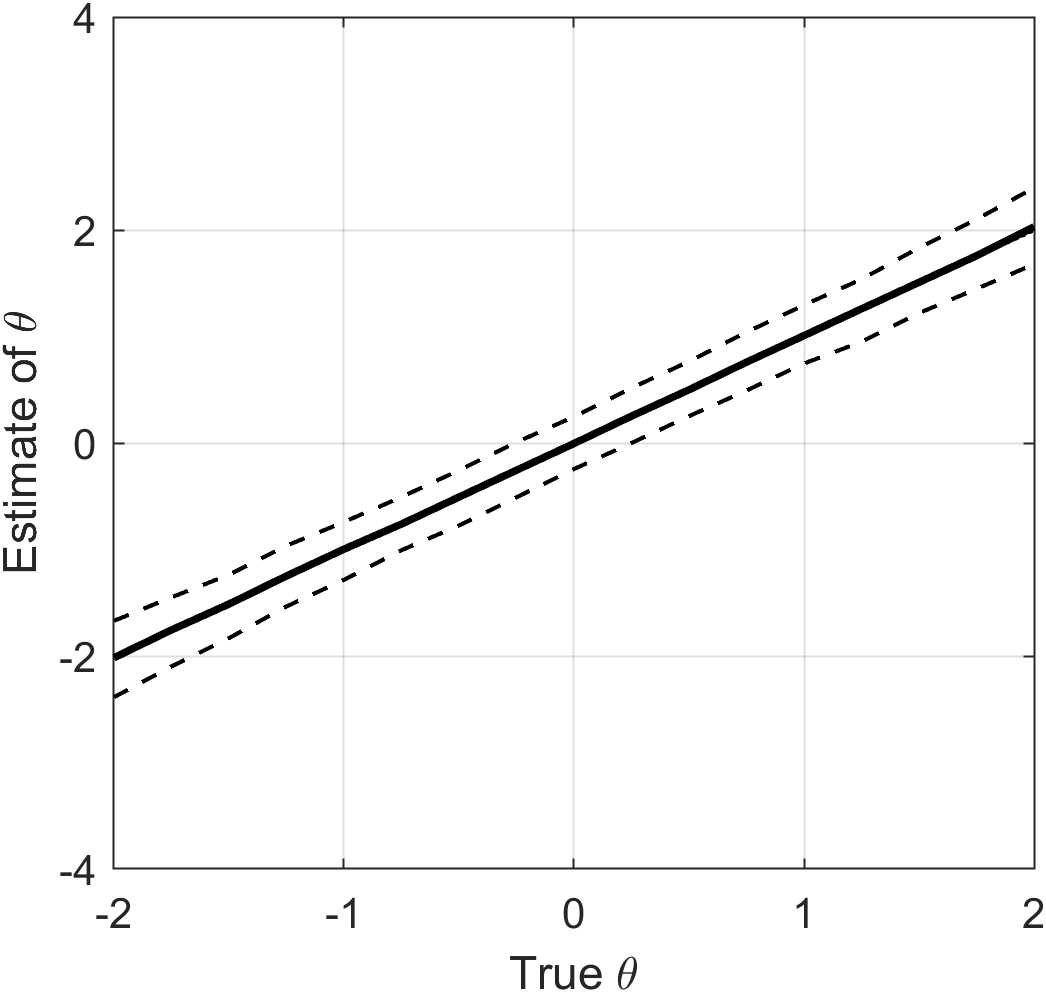}
    \end{tabular}\\
	\caption{\small Orthogonal estimators in a logit model \label{fig_mc_logit_2}}
	\floatfoot{\small Notes: see notes to Figure \ref{fig_mc_logit}.}
\end{figure}

 \paragraph{Example \ref{example_3} (continued).}
Consider now the nonlinear regression model (\ref{eq_nonlin_AKM}). Its log-likelihood is
$$
\ln f(y\mid x;\theta,\eta)=\frac{1}{\sigma^2}\,y^{\top}m(x,\beta,\eta)-\frac{1}{2\sigma^{2}}\left(y^{\top}y+m(x,\beta,\eta)^{\top}m(x,\beta,\eta)\right)-\frac{n}{2}\ln\left(2\pi\sigma^{2}\right),
$$
where recall that $\theta=(\beta^{\top},\sigma^2)^{\top}$. Viewed as a family in $\eta$ at fixed $\theta$, this is a Gaussian location family whose natural parameter is $m(x,\beta,\eta)/\sigma^2$ and whose sufficient statistic is $y$ itself. Since $\eta$ is $d_\eta$-dimensional while $y$ is $n$-dimensional, the natural parameter belongs to a $d_\eta$-dimensional submanifold of $\mathbb{R}^n$, and unless $m$ is linear in $\eta$ the model is a \emph{curved} exponential family in the sense of \citet{efron1975defining}.\footnote{In our context with covariates and nuisance parameters, a member of the curved exponential family takes the form
\begin{equation*}
	f(y\mid x;\theta,\eta)=a_{\theta}(y,x)\,b_{\theta}(x,\eta)\exp\left(s(y,x)^{\top}\tau_{\theta}(x,\eta)\right),
\end{equation*}
where $s(y,x)$ is the sufficient statistic and $\tau_{\theta}(x,\eta)$ the natural parameter.} This is the difference with Examples \ref{example_1} and \ref{example_2}: there the sufficient statistic $s$ reduced the data to $d_\eta$ dimensions, whereas here no reduction takes place.

Writing $z=(y-m)/\sigma$ for the vector of standardized residuals, the score for the nuisance parameter is
\begin{equation*}
	w_1=\nabla_{\eta}\ln f=\frac{1}{\sigma}\,\nabla_{\eta}m(x,\beta,\eta)z ,
\end{equation*}
and the higher-order elements of the Bhattacharyya basis follow from the chain rule. Indeed, the components of $z$ being independent standard normal, the derivatives of the Gaussian location family with respect to $m$ deliver products of Hermite polynomials: for a multi-index $j=(j_1,\ldots,j_n)$,
\begin{equation*}
	\frac{\nabla_{m}^{j}f(y\mid x;\theta,\eta)}{f(y\mid x;\theta,\eta)}=\frac{1}{\sigma^{|j|}}\prod_{i=1}^{n}H_{j_i}(z_i),\qquad |j|=\sum_{i=1}^{n}j_i,
\end{equation*}
where $H_k$ denotes the (probabilitists') Hermite polynomial of degree $k$ and $\nabla_{m}^{j}$ the corresponding mixed partial derivative with respect to $m$. Composing with $\eta\mapsto m(x,\beta,\eta)$ by the multivariate Fa\`a di Bruno formula \citep{constantine1996multivariate} then expresses $w_j$ as a combination of such products of total degree at most $j$, with coefficients built from the derivatives of $m$ up to order $j$. When $n=1$ and $\eta$ is scalar, for instance,
\begin{equation*}
	w_2=\frac{\nabla^2_{\eta}m}{\sigma}\,H_1(z)+\frac{\left(\nabla_{\eta}m\right)^2}{\sigma^{2}}\,H_2(z).
\end{equation*}

The projection (\ref{eq_q_proj}) is available in closed form despite the absence of a sufficient statistic. The reason is the orthogonality of the Hermite polynomials under the Gaussian reference measure,
$$
\mathbb{E}_{x;\theta,\eta}\left[\prod_{i=1}^{n}H_{j_i}(z_i)\,\prod_{i=1}^{n}H_{k_i}(z_i)\right]=\boldsymbol{1}\{j=k\}\prod_{i=1}^{n}j_i! ,
$$
which delivers the Gram matrix of $W_q$ and the cross moments with $u$ as explicit functions of the derivatives of $m$. \citet{bonhomme2024neyman} implement this approach to estimate the model of team production of \citet{AhmadpoodJones2019} using a sample of researchers in economics.

\subsection{The large-$q$ limit: connection to full orthogonality\label{sec_higher_limit}}

Full orthogonality implies orthogonality to every order $q\geq 1$ (by differentiating $\mathbb{E}_{x;\theta,\eta}[u^*(y,x,\theta,\eta')]=0$ in $\eta'$ under the integral sign). We now show that the converse is also true, provided that an analyticity condition holds.\footnote{\citet{small2011hilbert} similarly invoke analyticity to show that the Bhattacharyya basis spans the set of conditional-mean zero functions of the sufficient statistic in exponential family models.} In this sense, order-$q$ orthogonality interpolates continuously between Neyman and full orthogonality, with the limit recovering the latter
under analyticity. 

To develop the argument, suppose ${\mathcal{H}}$ is open and connected, $u^{*}(y,x,\theta,\eta)$ is orthogonal to every order $q\geq 1$, and the function
$$G_{x;\theta,\eta}:\eta'\mapsto\mathbb{E}_{x;\theta,\eta}\!\left[u^{*}(y,x,\theta,\eta')\right]$$ is real-analytic on ${\mathcal{H}}$, for all $\theta$, $\eta$, almost every $x$. Then, the entire Taylor expansion of $G_{x;\theta,\eta}$ around $\eta$ has vanishing
coefficients, so $G_{x;\theta,\eta}$ is identically zero on $\mathcal{H}$. This implies that, for all $\theta$, $\eta$, and almost every $x$,
\begin{equation*}
G_{x;\theta,\eta}(\eta')=	\int u^{*}(y,x,\theta,\eta')\,f(y\mid x;\theta,\eta)\,dy=0
\qquad \text{for all } \eta'\in\mathcal{H}.
\end{equation*}
It follows that $u^*$ is fully orthogonal.

\paragraph{Example \ref{example_1} (continued).}
In the panel logit model, because $s$ takes only the $T+1$ values $\{0,1,\ldots,T\}$, polynomials of degree at most $T$ already exhaust all functions of $s$. Since $w_j$ has degree $j$ in $s$, the functions $w_1,\ldots,w_T$ are linearly independent while $w_{T+1}$ is a combination of them. The construction thus terminates at $q=T$: the Gram matrix of $W_q$ is singular for $q>T$, and at $q=T$ the projection space contains all mean-zero functions of $s$, so that, by (\ref{eq_logit_full}),
$$
u^*_T=u-\mathbb{E}_{x;\theta,\eta}\left[u\mid s\right]=u^*_\infty.
$$
When $u=\nabla_{\theta}\ln f$, the factorization $\ln f(y\mid x;\theta,\eta)=\ln f(y\mid s,x;\theta)+\ln f(s\mid x;\theta,\eta)$ gives $\mathbb{E}_{x;\theta,\eta}[u\mid s]=\nabla_{\theta}\ln f(s\mid x;\theta,\eta)$, so that $u^*_T$ is the conditional-logit score. The sequence $u_1^*,\ldots,u_T^*$ thus interpolates between the variance-weighted demeaning of (\ref{eq_logit_orth}) and the fully-orthogonal moment $u^*_\infty$. 

This property is consistent with the Monte Carlo exercise reported in Figure \ref{fig_mc_logit_2}. In that setting, we used $T_{\rm est}=3$ periods for estimation after having used two periods for sample splitting. In the simulations, 
$u_3^*$ and the conditional-logit score $u_{\infty}^*$ coincide up to the tolerance of the root-finding algorithm.

The same argument applies to the general form (\ref{eq_plrlogit}). There $s=x_2^{\top}y$ is multivariate and the basis consists of the centered monomials in $s$ of degree $j$. Since $s$ has finite support, polynomials of sufficiently high degree exhaust the functions on that support (a bound on the required degree is the number of support points of $s$ minus 1). However, when the map $y\mapsto x_2^{\top}y$ is injective, the limiting moment is $u^*_{\infty}=0$.

We illustrate this last point with simulations based on a panel logit model with two scalar covariates and a heterogeneous coefficient, 
\begin{equation}
	y_{it}=\boldsymbol{1}\left\{\theta_0 x_{it}+\eta_{i0}z_{it}+\varepsilon_{it}\geq 0\right\},\,\, \varepsilon_{it}\mid x_{i1},\ldots,x_{iT},z_{i1},\ldots,z_{iT} \sim\ \text{i.i.d.~Logistic},
	\label{eq_logit_hetero}
\end{equation}
where $x_{it}$ is binary and $z_{it}$ is continuously distributed.

Figure \ref{fig_mc_logit_2_z} shows two orthogonal estimators, to order $q=2$ (left) and $q=3$ (right). The $q=2$ estimator is well-behaved and shows little bias. In contrast, the $q=3$ estimator exhibits bias and much wider confidence bands. In fact, the estimator fails to exist in approximately a quarter of the simulations. This contrasts sharply with Figure \ref{fig_mc_logit_2}, in the standard panel logit case, where the $q=3$ estimator was numerically identical to the fully-orthogonal conditional logit estimator. The model with a heterogeneous coefficient of a continuous covariate is a setting where the map $y_i\mapsto \sum_t z_{it}y_{it}$ is injective, and there is thus no non-trivial fully-orthogonal moment in this model. The simulations show that, in such cases, increasing the orthogonality order can have a high cost in terms of variability.

\begin{figure}[h!]
	\centering
    \begin{tabular}{cc}
 $q=2$ & $q=3$\\
	\includegraphics[width=.5\textwidth]{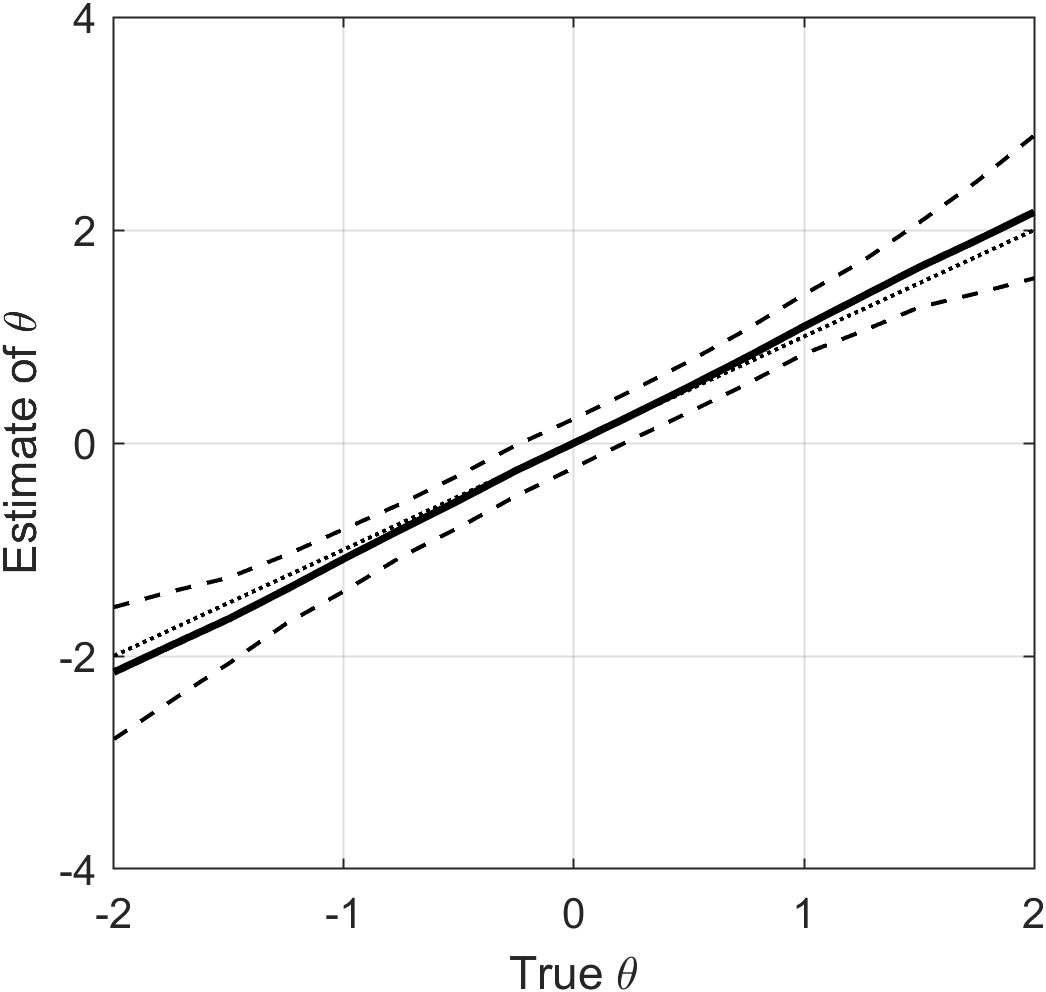}&\includegraphics[width=.5\textwidth]{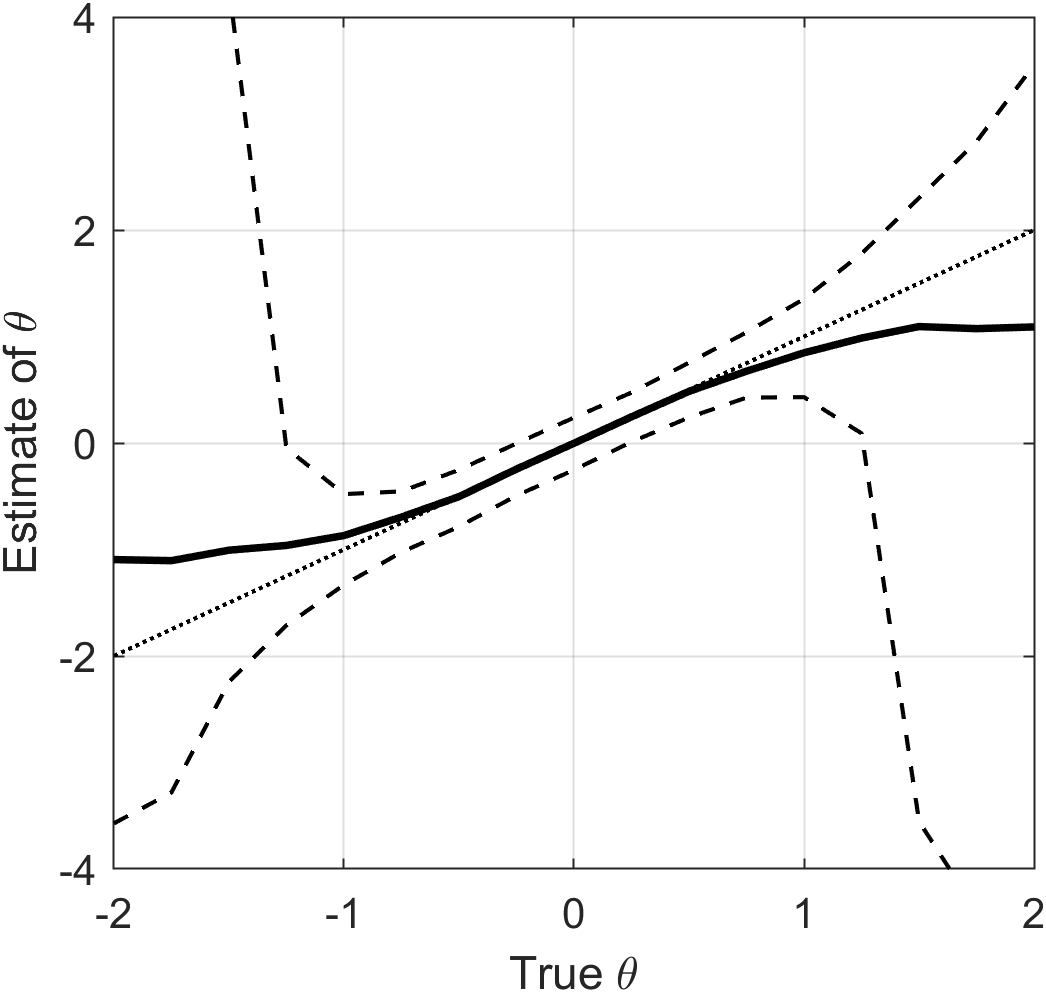}
    \end{tabular}\\
	\caption{\small Orthogonal estimators in a logit model with a heterogeneous coefficient \label{fig_mc_logit_2_z}}
	\floatfoot{\small Notes: see notes to Figure \ref{fig_mc_logit}. The results are computed based on the instances where the estimator exists (100\% of cases for $q=2$, approximately 75\% of cases for $q=3$).}
\end{figure}

\section{Other target parameters\label{sec_average}}
\setcounter{equation}{0}

In many applications the researcher is interested in a parameter $\mu_0\in{\cal{M}}$ that satisfies the moment condition
\begin{equation}
	\mathbb{E}\!\left[u(y,x,\theta_0,\eta_0,\mu_0)\right]=0,
	\label{eq_mu_target}
\end{equation}
and thus is a functional of $\theta_0$, $\eta_0$, and the distribution of $x$. Leading examples are average partial effects, average elasticities, and counterfactual policy effects in structural models. The three notions of orthogonality can be adapted to handle these quantities. 

\subsection{Orthogonality}

The key difference with the case discussed in the earlier sections is that $u$ is not unbiased at all parameter values, that is, 
$$\mathbb{E}_{x;\theta,\eta}\left[u(y,x,\theta,\eta,\mu)\right]$$
is not identically zero: only its expectation over $x$ is zero when evaluated at the true values $\theta=\theta_0$, $\eta=\eta_0$, $\mu=\mu_0$.

For all three notions, we replace the unbiasedness requirement on the orthogonal moment by 
\begin{align}
	&\mathbb{E}_{x;\theta,\eta}\left[u^*(y,x,\theta,\eta,\mu)\right]=\mathbb{E}_{x;\theta,\eta}\left[u(y,x,\theta,\eta,\mu)\right]\notag\\& \quad\quad\quad\quad\quad\text{for all }\theta\in\Theta,\ \eta\in{\cal{H}},\  \mu\in{\cal{M}},\ \text{ almost every } x,
	\label{eq_mom_star_lik_target}
\end{align}
and we will require (\ref{eq_mom_star_lik_target}) for all $u^*\in \{u_1^*,u_q^*,u_{\infty}^*\}$.

With this change, Neyman orthogonality is as before, except for the presence of $\mu$ and for the replacement of unbiasedness by (\ref{eq_mom_star_lik_target}). That is, we say that $u_1^*$ is Neyman orthogonal if
\begin{align}
	&\mathbb{E}_{x;\theta,\eta}\left[\nabla_{\eta}u^*_{1}(y,x,\theta,\eta,\mu)\right]=0\,\, \text{for all }\theta\in\Theta,\ \eta\in{\cal{H}},\  \mu\in{\cal{M}},\ \text{ almost every } x,
	\label{eq_neyman_target}
\end{align}
and $u^*=u^*_1$ satisfies (\ref{eq_mom_star_lik_target}).

Likewise, for all $q\geq 1$, we say that $u_q^*$ is orthogonal to order $q$ if 
\begin{align}
	&\mathbb{E}_{x;\theta,\eta}\left[\nabla_{\eta}^ju^*_{q}(y,x,\theta,\eta,\mu)\right]=0\notag\\&\quad\quad\quad\quad\quad\text{for all }j=1,\ldots,q,\ \theta\in\Theta,\ \eta\in{\cal{H}},\  \mu\in{\cal{M}},\ \text{ almost every } x,
	\label{eq_orderq_target}
\end{align} 
and $u^*=u^*_q$ satisfies (\ref{eq_mom_star_lik_target}).

Finally, we say that $u^*_{\infty}$ is fully orthogonal if 
\begin{align}
	&\mathbb{E}_{x;\theta,\eta}\left[u^*_{\infty}(y,x,\theta,\eta',\mu)\right] \text{ does not depend on }\eta'\in{\cal{H}}\notag\\ &\quad\quad\quad\quad\quad\text{for all }\theta\in\Theta,\ \eta\in{\cal{H}},\  \mu\in{\cal{M}},\ \text{ almost every } x,
	\label{eq_fd_target}
\end{align}
and $u^*=u^*_{\infty}$ satisfies (\ref{eq_mom_star_lik_target}).

Under full orthogonality, plugging in any arbitrary value for $\eta'$, using (\ref{eq_fd_target}) and (\ref{eq_mom_star_lik_target}), and noting that $u$ is unbiased at the truth, we obtain
\begin{align*}
\mathbb{E}\left[u^*_{\infty}(y,x,\theta_0,\eta',\mu_0)\right]
&=\mathbb{E}\left[\mathbb{E}_{x;\theta_0,\eta_0}\left[u^*_{\infty}(y,x,\theta_0,\eta',\mu_0)\right]\right]\\
&=\mathbb{E}\left[\mathbb{E}_{x;\theta_0,\eta_0}\left[u^*_{\infty}(y,x,\theta_0,\eta_0,\mu_0)\right]\right]\\
&=\mathbb{E}\left[\mathbb{E}_{x;\theta_0,\eta_0}\left[u(y,x,\theta_0,\eta_0,\mu_0)\right]\right]=0,
\end{align*}
where the outer expectations are taken with respect to the distribution of $x$. Given a consistent estimate of $\theta_0$, $\mu_0$ can then be consistently estimated using standard GMM.

\subsection{The Neyman projection}

The orthogonal constructions take a different form in the case of the target parameter $\mu_0$. Starting with Neyman orthogonality, note that differentiating (\ref{eq_mom_star_lik_target}) in $\eta$ for $u^*=u_1^*$ implies
\begin{align*}&\mathbb{E}_{x;\theta,\eta}\left[\nabla_\eta u^*_{1}(y,x,\theta,\eta,\mu)\right]+\mathbb{E}_{x;\theta,\eta}\left[w_1(y,x,\theta,\eta)\,u^*_{1}(y,x,\theta,\eta,\mu)^{\top}\right]\\&=\nabla_{\eta}\mathbb{E}_{x;\theta,\eta}\left[u(y,x,\theta,\eta,\mu)\right],\end{align*}
where $w_1$ is the score with respect to $\eta$, and recall that, with our definition, the matrix $\nabla_{\eta}\mathbb{E}_{x;\theta,\eta}\left[u(y,x,\theta,\eta,\mu)\right]$ is $d_{\eta}\times d_u$.

Hence, Neyman orthogonality is equivalent to (\ref{eq_mom_star_lik_target}) at $u^*=u_1^*$, and
\begin{align}
	&\mathbb{E}_{x;\theta,\eta}\left[w_1(y, x,\theta,\eta)\,u^*_1(y,x,\theta,\eta,\mu)^{\top}\right]=\nabla_{\eta}\mathbb{E}_{x;\theta,\eta}\left[u(y,x,\theta,\eta,\mu)\right]\notag\\& \quad\quad\quad\quad\quad\text{for all }\theta\in\Theta,\ \eta\in{\cal{H}},\ \mu\in{\cal{M}},\ \text{almost every }x.
	\label{eq_neyman_score_target}
\end{align}

Conditions (\ref{eq_mom_star_lik_target}) and (\ref{eq_neyman_score_target}) define a set of affine restrictions on $u_1^*$. Denote this set by ${\cal{T}}_{x;\theta,\eta,\mu}^{1}$. Given an original moment function $u$, we construct a Neyman-orthogonal $u_1^*$ as
\begin{equation}
	u^*_1(\cdot,x,\theta,\eta,\mu)=\operatorname{Proj}_{x;\theta,\eta}\left[u(\cdot,x,\theta,\eta,\mu)\,\middle\|\, {\cal{T}}_{x;\theta,\eta,\mu}^{1}\right],
	\label{eq_neyman_proj_target}
\end{equation}
where $\operatorname{Proj}_{x;\theta,\eta}\left[v\,\middle\|\,{\cal S}\right]$ is given by (\ref{eq_proj}).

Since $\eta$ is finite-dimensional, the projection is finite-dimensional as well, and $u_1^*$ in (\ref{eq_neyman_proj_target}) takes the explicit form
\begin{equation}
	u^*_1(y,x,\theta,\eta,\mu)=u(y,x,\theta,\eta,\mu)-\Lambda_1(x;\theta,\eta,\mu)\, w_1(y, x ,\theta,\eta),
	\label{eq_neyman_construction_target}
\end{equation}
with projection coefficient
\begin{align*}
	&\Lambda_1(x;\theta,\eta,\mu)=\left(\mathbb{E}_{x;\theta,\eta}\left[u(y,x,\theta,\eta,\mu)\, w_1(y,x,\theta,\eta)^{\top}\right]-b_1(x;\theta,\eta,\mu)^{\top}\right) \\&\quad\quad\quad\quad\quad\quad\quad\quad\quad\times \mathbb{E}_{x;\theta,\eta}\left[w_1(y,x,\theta,\eta)\, w_1(y,x,\theta,\eta)^{\top}\right]^{-1},
\end{align*}
where
$$
b_1(x;\theta,\eta,\mu)=\nabla_{\eta}\,\mathbb{E}_{x;\theta,\eta}\left[u(y,x,\theta,\eta,\mu)\right]
$$
is the sensitivity of the expected moment to the nuisance parameter. Indeed, $u_1^*$ in (\ref{eq_neyman_construction_target}) satisfies (\ref{eq_mom_star_lik_target}), since $w_1$ has mean zero, and it satisfies (\ref{eq_neyman_score_target}) by construction of $\Lambda_1$. Comparing with (\ref{eq_neyman_construction}), the only difference is the subtraction of $b_1$ from the covariance between $u$ and the nuisance score.

\subsection{The order-$q$ projection}

The same argument applies at higher orders. Proceeding as in Section \ref{sec_higher_constr}, but starting from (\ref{eq_mom_star_lik_target}) in place of unbiasedness at all parameter values, orthogonality to order $q$ is equivalent to (\ref{eq_mom_star_lik_target}) at $u^*=u^*_q$ and
\begin{align}
	&\mathbb{E}_{x;\theta,\eta}\left[w_j(y, x,\theta,\eta)\,u^*_q(y,x,\theta,\eta,\mu)^{\top}\right]=\nabla_{\eta}^j\,\mathbb{E}_{x;\theta,\eta}\left[u(y,x,\theta,\eta,\mu)\right]\notag\\& \quad\quad\quad\quad\quad\text{for all }j=1,\ldots,q,\ \theta\in\Theta,\ \eta\in{\cal{H}},\ \mu\in{\cal{M}},\ \text{almost every }x,
	\label{eq_qorth_score_target}
\end{align}
where the right-hand side is a total derivative, with $\eta$ entering both the expectation and $u$.

Conditions (\ref{eq_mom_star_lik_target}) and (\ref{eq_qorth_score_target}) again define a set of affine restrictions on $u_q^*$, which we denote by ${\cal{T}}_{x;\theta,\eta,\mu}^{q}$, and we construct
\begin{equation}
	u^*_q(\cdot,x,\theta,\eta,\mu)=\operatorname{Proj}_{x;\theta,\eta}\left[u(\cdot,x,\theta,\eta,\mu)\,\middle\|\, {\cal{T}}_{x;\theta,\eta,\mu}^{q}\right].
	\label{eq_q_proj_target}
\end{equation}

Stacking $W_q=(w_1^{\top},\ldots,w_q^{\top})^{\top}$ and $b_q=\left(b_1^{\top},\ldots,b_q^{\top}\right)^{\top}$, where
$$
b_j(x;\theta,\eta,\mu)=\nabla_{\eta}^j\,\mathbb{E}_{x;\theta,\eta}\left[u(y,x,\theta,\eta,\mu)\right],
$$
the projection (\ref{eq_q_proj_target}) takes the explicit form
\begin{equation}
	u^*_q(y,x,\theta,\eta,\mu)=u(y,x,\theta,\eta,\mu)-\Lambda_q(x;\theta,\eta,\mu)\, W_q(y, x ,\theta,\eta),
	\label{eq_qorth_explicit_target}
\end{equation}
with projection coefficient
\begin{align*}
	&\Lambda_q(x;\theta,\eta,\mu)=\left(\mathbb{E}_{x;\theta,\eta}\left[u(y,x,\theta,\eta,\mu)\, W_q(y,x,\theta,\eta)^{\top}\right]-b_q(x;\theta,\eta,\mu)^{\top}\right) \\&\quad\quad\quad\quad\quad\quad\quad\quad\quad\times \mathbb{E}_{x;\theta,\eta}\left[W_q(y,x,\theta,\eta)\, W_q(y,x,\theta,\eta)^{\top}\right]^{-1}.
\end{align*}
Similarly to the case $q=1$, (\ref{eq_mom_star_lik_target}) holds at $u^*=u^*_q$ because the $w_j$ have mean zero, and (\ref{eq_qorth_score_target}) holds by construction of $\Lambda_q$. Comparing with (\ref{eq_qorth_explicit}), the construction differs only through the subtraction of $b_q$. Expression (\ref{eq_qorth_explicit_target}) is Theorem 1 in \citet{bonhomme2024neyman}.

\subsection{The fully-orthogonal projection}

Since (\ref{eq_fd_target}) holds at every value of $\eta'$, we may fix it at any convenient value and simply write $u^*_{\infty}(y,x,\theta,\mu)$, as in Section \ref{sec_full_def}. The condition then reads, for all $\theta$, $\eta$, and $x$,
$$
\mathbb{E}_{x;\theta,\eta}\left[u^{*}_{\infty}(y,x,\theta,\mu)\right]=b_0(x;\theta,\eta,\mu)^{\top}\qquad\text{for all }\eta\in{\cal{H}} ,
$$
where $b_0(x;\theta,\eta,\mu)=\mathbb{E}_{x;\theta,\eta}[u(y,x,\theta,\eta,\mu)^{\top}]$. Fixing a reference value $\eta\in{\cal{H}}$ and applying the change of measure of Section \ref{sec_fd}, we obtain that (\ref{eq_fd_target}) is equivalent to
\begin{align}
	&\mathbb{E}_{x;\theta,\eta}\left[u^*_{\infty}(y,x,\theta,\mu)\,r_{\widetilde\eta}(y,x,\theta,\eta)\right]=b_0(x;\theta,\widetilde\eta,\mu)^{\top}\notag\\&\quad\quad\quad\quad\quad\text{for all }\theta\in\Theta,\ \eta\in{\cal{H}},\ \widetilde\eta\in{\cal{H}},\ \mu\in{\cal{M}},\ \text{almost every }x.
	\label{eq_full_score_target}
\end{align}

Condition (\ref{eq_full_score_target}) defines an affine subset of $L^2\big(f(\cdot\mid x;\theta,\eta)\big)$, which we denote by ${\cal{T}}^{\infty}_{x;\theta,\eta,\mu}$. As before, we construct
\begin{equation}
	u^*_\infty(\cdot,x,\theta,\eta,\mu)=\operatorname{Proj}_{x;\theta,\eta}\left[u(\cdot,x,\theta,\eta,\mu)\,\middle\|\, {\cal{T}}^{\infty}_{x;\theta,\eta,\mu}\right] ,
	\label{eq_fd_proj_target}
\end{equation}
the closest fully-orthogonal moment to $u$ in $L^2\big(f(\cdot\mid x;\theta,\eta)\big)$. As in Section \ref{sec_fd}, the feasible set consists of functions that do not involve $\eta$, whereas the projected moment does. When $u$ is unbiased at all parameter values we have $b_0= 0$, and (\ref{eq_fd_proj_target}) coincides with the residual construction (\ref{eq_fd_proj}) of Section \ref{sec_fd}.

Two differences with the other projections are worth noting. First, the fully-orthogonal projection is again infinite-dimensional; Appendix~\ref{App_variational} gives the corresponding operator equation. Second, ${\cal{T}}^{\infty}_{x;\theta,\eta,\mu}$ may be empty, in which case no fully-orthogonal moment exists for the target $\mu_0$, as we illustrate in Example \ref{example_1} next.

\subsection{Examples}

\paragraph{Example \ref{example_1} (continued).}
Consider the panel logit model (\ref{eq_logit}) with a scalar continuous covariate, suppress the unit index as before, and consider the average partial effect of $x_t$
\begin{equation}
	\mu_0=\mathbb{E}\left[\frac{1}{T}\sum_{t=1}^{T}\nabla_{x}\Pr\left(y_{t}=1\mid x_{t}=x\right)\right]=\mathbb{E}\left[\frac{1}{T}\sum_{t=1}^{T}\theta_0\,\pi_{t}(1-\pi_{t})\right].
	\label{eq_ape}
\end{equation}
A unit's contribution to the moment condition is
$$
u(y,x,\theta,\eta,\mu)=\frac{\theta}{T}\sum_{t=1}^{T}v_t(\theta,\eta)-\mu ,\qquad v_t=\pi_t(1-\pi_t).
$$
This moment does not depend on the outcome, so that $\mathbb{E}_{x;\theta,\eta}[u\,w_j^{\top}]=0$ for every $j$ and the construction only depends on the sensitivities $b_j$. Those have a compact form. Recall from Section \ref{sec_higher_ex} that $\nabla^j_{\eta}\kappa$ is the $j$th cumulant $\kappa_j$ of the sufficient statistic $s=\sum_t y_t$, and note that $\sum_t v_t=\kappa_2$. Hence
$$
b_0=\frac{\theta}{T}\,\kappa_2-\mu,\qquad b_j=\frac{\theta}{T}\,\kappa_{j+2},\qquad j\geq1 .
$$

The orthogonal moments follow. With $\mathbb{E}_{x;\theta,\eta}[w_1^2]=\kappa_2$, the Neyman construction (\ref{eq_neyman_construction_target}) gives
\begin{equation*}
	u^*_1=u+\frac{\theta}{T}\,\frac{\kappa_3}{\kappa_2}\,w_1 ,\qquad w_1=s-\kappa_1 ,
\end{equation*}
and, using $w_2=(s-\kappa_1)^2-\kappa_2$ together with $\mathbb{E}[w_1w_2]=\kappa_3$ and $\mathbb{E}[w_2^2]=\kappa_4+2\kappa_2^2$, the order-$2$ construction (\ref{eq_qorth_explicit_target}) gives
\begin{equation*}
	u^*_2=u+\frac{\theta}{T}\left(\kappa_3,\ \kappa_4\right)
	\begin{pmatrix}\kappa_2 & \kappa_3\\ \kappa_3 & \kappa_4+2\kappa_2^{2}\end{pmatrix}^{-1}
	\begin{pmatrix}w_1\\ w_2\end{pmatrix} .
\end{equation*}

Fully-orthogonal moments, in contrast, do not exist here. Fix $x$ and $\theta$, and write a candidate as $u^*_\infty=g(s,x,\theta,\mu)+h$, with $\mathbb{E}_{x;\theta,\eta}[h\mid s]=0$. By sufficiency, the conditional distribution of $y$ given $s$ does not depend on $\eta$, so $h$ contributes nothing at any value of the fixed effect and
$$
\mathbb{E}_{x;\theta,\eta}\left[u^*_\infty(y,x,\theta,\mu)\right]=\sum_{k=0}^{T}g(k,x,\theta,\mu)\,\Pr\left(s=k\mid x;\theta,\eta\right).
$$
Full orthogonality thus requires the $T+1$ numbers $g(0,x,\theta,\mu)$, ..., $g(T,x,\theta,\mu)$ to satisfy
$$\sum_{k=0}^{T} g(k,x,\theta,\mu)\,\Pr(s=k\mid x;\theta,\eta)=b_0(x;\theta,\eta,\mu),\qquad \text{for all }\eta\in\mathcal{H},$$
which is a finite-dimensional system facing a continuum of restrictions. For $T=2$, clearing denominators and matching powers of $e^\eta$ leaves the restriction
$$
\frac{\theta}{2}\left(e^{x_1\theta}-e^{x_2\theta}\right)^2e^{2\eta}=0,
$$
which holds only when $x_1=x_2$ or $\theta=0$. 

\paragraph{Example \ref{example_3} (continued).}
As a special case of model (\ref{eq_nonlin_AKM}), consider the linear model
\begin{equation}
	y=x\eta_0 +\sigma_0 \varepsilon,\qquad \varepsilon\mid x\ \sim\ {\cal N}(0,I_n).
	\label{eq_AKM}
\end{equation}
A leading application is the two-way model of \citet{AbowdKramarzMargolis1999}, in which $y$ collects log wages, $x$ encodes worker and firm identities, and $\eta_0$ the corresponding effects. Besides $\sigma^2_0$, interest centers on variance components: quadratic forms $\mu_0=\eta_0^{\top}Q\eta_0$ for a symmetric $Q$, measuring the dispersion of worker and firm effects and their covariance (see \citealp{bonhomme2026users} for a survey). Here we take $\sigma^2$ as known, and start with the original moment function $u(y,x,\theta,\eta,\mu)=\eta^{\top}Q\eta-\mu$. We provide detailed derivations in Appendix \ref{App_lin_network}.

Since $u$ does not depend on the outcome, the construction is governed by the sensitivities alone. Here the construction terminates at $q=2$, since the $w_j$ are uncorrelated and
$$
b_0=\eta^{\top}Q\eta-\mu,\qquad b_1=2Q\eta,\qquad b_2=2\operatorname{vec}\left(Q\right),\qquad b_j=0\ \text{ for }j\geq3 .
$$
With $w_1=\sigma^{-2}x^{\top}(y-x\eta)$ and $\mathbb{E}_{x;\theta,\eta}[w_1w_1^{\top}]=\sigma^{-2}x^{\top}x$, the Neyman construction gives
\begin{equation}
	u_1^{*}=\eta^{\top}Q\eta-\mu+2\,\eta^{\top}Q(x^{\top}x)^{-1}x^{\top}(y-x\eta),
	\label{eq_akm_u1}
\end{equation}
which still involves $\eta$. This moment is moreover of no practical use when $\widehat\eta$ is the least-squares estimator computed on the same sample: then $x^{\top}(y-x\widehat\eta)=0$, the correction vanishes, and $u_1^*$ reduces to the uncorrected plug-in $\widehat\eta^{\top}Q\widehat\eta$.

Imposing second-order orthogonality changes the picture. Adding $$w_2=\operatorname{vec}\left(-\sigma^{-2}x^{\top}x+\sigma^{-4}x^{\top}(y-x\eta)(y-x\eta)^{\top}x\right)$$ to the basis, the order-two moment is
\begin{equation}
	u_2^{*}=y^{\top}x(x^{\top}x)^{-1}Q(x^{\top}x)^{-1}x^{\top}y-\sigma^2\operatorname{tr}\left(Q(x^{\top}x)^{-1}\right)-\mu,
	\label{eq_akm_u2}
\end{equation}
which no longer depends on $\eta$ and is unbiased at every value of it. Moreover, $u_2^*$ is fully orthogonal. Note the contrast with Example \ref{example_1}: while no fully-orthogonal moment exists for the average effect in the panel logit model with a continuous covariate, in model (\ref{eq_AKM}) quadratic forms in $\eta$ admit fully-orthogonal moments. We recognize the trace-corrected variance-component expression of \citet{AndrewsMartynGillSchankUpward}, generalized by \citet{KlineSaggioSoelvsten2020} to the non-Gaussian, heteroskedastic case.\footnote{A second-order projection also provides an unbiased estimator of $\sigma^2_0$ itself, which we have taken as known. Writing $M_x=I_n-x(x^{\top}x)^{-1}x^{\top}$ and taking $u=\nabla_{\sigma^2}\ln f=-\frac{n}{2\sigma^2}+\frac{1}{2\sigma^4}(y-x\eta)^{\top}(y-x\eta)$, the moment
$$
u^*_2=\frac{1}{2\sigma^4}\left(y^{\top}M_xy-\left(n-d_\eta\right)\sigma^2\right)
$$
has mean zero at every $\eta\in{\cal{H}}$, since $M_x y=\sigma M_x\varepsilon$ does not involve $\eta$ and $\mathbb{E}[\varepsilon^{\top}M_x\varepsilon]=\operatorname{tr}(M_x)=n-d_\eta$. It is thus fully orthogonal. Solving $u_2^*=0$ gives the degrees-of-freedom corrected variance estimator $y^{\top}M_xy/(n-d_\eta)$. See Appendix \ref{App_lin_network}.\label{ftnote_ex3}}

A second special case is the nonlinear model
\begin{equation}
	y=x_{1}\eta_0+\beta_0\,x_{2}(\eta_0\otimes\eta_0)+\sigma_0\varepsilon,
	\quad\varepsilon\mid x\sim\text{iid }{\cal N}(0,I_{n}),
	\label{eq_Tukey}
\end{equation}
where $x=(x_{1},x_{2})$, $\otimes$ denotes the Kronecker product, and $\beta_0$ is scalar. This model contains as a special case the Tukey generalization of \citet{AbowdKramarzMargolis1999} studied in \citet{crippa2025identification}, here for Gaussian errors, and reduces to (\ref{eq_AKM}) when the complementarity parameter $\beta_0$ equals zero.

Let us focus on the variance component $\mu_0=\eta_0^{\top}Q\eta_0$. Since the sensitivities $b_j$ depend on the target alone, they are as above. However, the Bhattacharyya basis changes relative to the linear model. Writing $m(\eta)=x_1\eta+\beta x_2(\eta\otimes\eta)$ and $z=(y-m(\eta))/\sigma$, the Jacobian
$$
M(\eta)=\left(\nabla_{\eta}m(\eta)\right)^{\top}=x_1+\beta\,x_2\left(I_{d_\eta}\otimes\eta+\eta\otimes I_{d_\eta}\right)
$$
now depends on $\eta$, and $\nabla^2_{\eta}m$ is a nonzero constant proportional to $\beta$. Following Example \ref{example_3} in Subsection \ref{sec_higher_ex}, $w_1=\sigma^{-1}M(\eta)^{\top}z$ and $w_2$ is the corresponding combination of first- and second-order Hermite polynomials in $z$, with coefficients built from $M(\eta)$ and $\nabla^2_{\eta}m$. The Gram matrix of $W_2$ and the cross moments follow in closed form from the orthogonality of the Hermite polynomials, so that the order-two moment $u_2^*$ in (\ref{eq_qorth_explicit_target}) is available in closed form, and the same holds at higher orders. Unlike (\ref{eq_akm_u2}), however, $u_2^*$ still depends on $\eta$ through $M(\eta)$, so that estimation requires a first-step estimate and sample splitting. \citet{crippa2025identification} constructs fully-orthogonal moments for $\beta_0$ in the special case of the Tukey model. Whether informative fully-orthogonal moments exist for quadratic forms $\mu_0$ in (\ref{eq_Tukey}) is left for future work.

\section{Conclusion\label{sec_conclusion}}
\setcounter{equation}{0}

We have described three constructions that make a moment function insensitive to the presence of nuisance parameters in a nonlinear model. The first is Neyman orthogonality, the second is orthogonality to order $q$, and the third is full orthogonality. All three apply the same principle. In each case the moment function is projected on a subspace of functions of the data determined by the likelihood, and the residual is retained. The three differ only in the choice of subspace, which is spanned by the nuisance score, by the first $q$ generalized scores, and by all the likelihood ratios, respectively. These subspaces are nested and, under an analyticity condition, orthogonality to every order implies full orthogonality, so that the three notions form a single increasing sequence of projections. We have also shown that the projection delivering full orthogonality is Neyman's own construction, applied to a model in which the distribution of the nuisance parameter is itself the nuisance parameter. Finally, we have shown that the projection approach extends naturally to target parameters that are functionals of the nuisance parameters and of the distribution of the covariates.

Our examples illustrate when each notion is available. In the panel Poisson model, Neyman orthogonality already delivers a fully-orthogonal moment. In the panel logit, taking $q=T$ returns the conditional-likelihood score. For average effects of a continuous covariate in that model no fully-orthogonal moment exists, and orthogonality to a finite order is the only route available among these constructions. Choosing the order $q$ in practice remains an open question.

Developing methods to construct orthogonal moments outside likelihood models is an important direction for future research. Neyman orthogonality extends naturally to this case, as shown by an active recent literature on double/debiased machine learning. Recent research has extended higher-order orthogonality to moment condition models without a likelihood structure \citep{chetverikov2026triple,bonhomme2026higher}. More work in this direction is needed to better understand the properties of estimators based on orthogonal moments in settings such as panels, networks, and non- and semiparametric estimation problems where nuisance parameters are high-dimensional.

\clearpage

\ifx\undefined\BySame
\newcommand{\BySame}{\leavevmode\rule[.5ex]{3em}{.5pt}\ }
\fi
\ifx\undefined\textsc
\newcommand{\textsc}[1]{{\sc #1}}
\newcommand{\emph}[1]{{\em #1\/}}
\let\tmpsmall\small
\renewcommand{\small}{\tmpsmall\sc}
\fi

\appendix
\renewcommand{\theequation}{\thesection.\arabic{equation}}

\clearpage

\begin{center}
{\Large APPENDIX}
\end{center}

\section{Details for the panel data logit model\label{App_binary}}
\setcounter{equation}{0}

\subsection{The second derivative of the expected Neyman-orthogonal moment\label{App_binary_G}}

To show (\ref{eq_logit_G2}), note that $\partial_{\eta'}\pi_t=v_t$ and $\partial_{\eta'}v_t=v_t(1-2\pi_t)$. Applying the quotient rule to $\overline{x}(\eta')=\sum_t v_tx_t/\sum_t v_t$ and collecting terms gives
\begin{equation*}
	\partial_{\eta'}\overline{x}(\eta')=\frac{\sum_{t=1}^{T} v_t\left(1-2\pi_t\right)\left(x_t-\overline{x}\right)}{\sum_{t=1}^{T} v_t}.
\end{equation*}
Differentiating (\ref{eq_logit_G}) twice with respect to $\eta'$ produces three terms. The first one carries $\pi_t(\eta)-\pi_t(\eta')$ as a factor and so vanishes at $\eta'=\eta$. The second one is $2\,\partial_{\eta'}\overline{x}(\eta')\sum_t v_t$, which by the previous display equals twice the right-hand side of (\ref{eq_logit_G2}). The third one is $-\sum_t\left(x_t-\overline{x}\right)\partial_{\eta'}v_t$, which equals minus the right-hand side of (\ref{eq_logit_G2}). Adding the last two terms gives the result. In turn, (\ref{eq_logit_G2_T2}) follows on substituting
$$
x_1-\overline{x}=\frac{v_2\left(x_1-x_2\right)}{v_1+v_2},\qquad x_2-\overline{x}=-\frac{v_1\left(x_1-x_2\right)}{v_1+v_2}
$$
into (\ref{eq_logit_G2}) and using $(1-2\pi_1)v_1v_2-(1-2\pi_2)v_1v_2=2v_1v_2(\pi_2-\pi_1)$.

\subsection{Moments of the sufficient statistic\label{App_binary_moments}}

Carrying out the projection (\ref{eq_logit_q}) in the panel logit model requires the moments of $s=\sum_{t=1}^{T}y_t$, which supply the coefficients of the $w_j$, together with the cross moments $\mathbb{E}_{x;\theta,\eta}[u\,s^{j}]$.

The cumulants of $s$ are obtained by adding up the per-period Bernoulli cumulants, and coincide with the derivatives $\nabla^j_{\eta}\kappa$ of the cumulant generating function of Section \ref{sec_higher_ex}. Writing $v_t=\pi_t(1-\pi_t)$, the four lowest are
\begin{equation*}
	\kappa_1=\sum_{t=1}^T \pi_t,\qquad \kappa_2=\sum_{t=1}^T v_t,\qquad \kappa_3=\sum_{t=1}^T v_t(1-2\pi_t),\qquad \kappa_4=\sum_{t=1}^T v_t(1-6v_t),
\end{equation*}
and the standard moment-cumulant relations turn these into the moments of $s$.

For the cross moments, decompose $s=y_t+s_{-t}$, where $s_{-t}=\sum_{r\neq t}y_r$ is independent of $y_t$, and use $y_t^k=y_t$ for $k\geq1$ to obtain
\begin{equation*}
	\mathbb{E}_{x;\theta,\eta}\left[(y_t-\pi_t)\,s^{j}\right]=v_t\sum_{k=1}^{j}\binom{j}{k}\,\mathbb{E}_{x;\theta,\eta}\left[s_{-t}^{\,j-k}\right],
\end{equation*}
the moments of $s_{-t}$ following from its cumulants in the same way. Summing over $t$ against $x_t$ delivers $\mathbb{E}_{x;\theta,\eta}[u\,s^j]$ for the score $u=\sum_t x_t(y_t-\pi_t)$.

\section{Construction of fully-orthogonal moments\label{App_variational}}
\setcounter{equation}{0}

\subsection{The target parameter $\theta$\label{App_var_theta}}

Given a moment function $u$, and for $x$, $\theta$, $\eta$ fixed, the projection residual $u_{\infty}^*$ in (\ref{eq_fd_proj}) solves 
\begin{align}
	u_{\infty}^*(\cdot,x,\theta,\eta)=&\, \underset{u^*}{\argmin}\, \mathbb{E}_{x;\theta,\eta}\left[\left\|u(y,x,\theta,\eta)-u^*(y)\right\|^2\right],\label{eq_FD1}\\
	& 	\, \text{s.t.}\quad \mathbb{E}_{x;\theta,\widetilde{\eta}}\left[u^*(y)\right]=0 \text{ for all }\widetilde{\eta}\in\mathcal{H}.\label{eq_FD2}
\end{align}

Components of $u$ decouple in the objective and in the constraints alike, so nothing is lost by treating a scalar component. The constraints are associated to a Lagrange multiplier $\lambda(\widetilde\eta)$. We then define the Lagrangian
\begin{equation}
	{\cal L}(u^*,\lambda)
	=\tfrac{1}{2}\,\mathbb{E}_{x;\theta,\eta}\left[\left(u-u^*\right)^2\right]
	-\int_{{\cal H}}\lambda(\widetilde\eta)\,
	\mathbb{E}_{x;\theta,\eta}\left[u^*\,r_{\widetilde\eta}\right]d\widetilde\eta,
	\label{eq_FD_lagrangian}
\end{equation}
each constraint having been expressed under the reference distribution using the change of measure underlying (\ref{eq_full_score}). 

Taken pointwise in $y$, the first-order condition delivers
\begin{equation}
	u_{\infty}^{*}(y,x,\theta,\eta)=u(y,x,\theta,\eta)
	+\int_{{\cal H}}\lambda(\widetilde\eta)\,
	\frac{f(y\mid x;\theta,\widetilde\eta)}{f(y\mid x;\theta,\eta)}\,d\widetilde\eta,
	\label{eq_FD_solution}
\end{equation}
so the program is solved by the original moment plus a correction drawn from ${\cal S}_{x;\theta,\eta}$. 

Finally, imposing the constraints, the Lagrange multiplier solves the following integral equation
\begin{equation}
		\int_{{\cal H}} K(\widetilde\eta',\widetilde\eta)\,\lambda(\widetilde\eta)\,d\widetilde\eta
	=-\,\mathbb{E}_{x;\theta,\widetilde\eta'}\left[u(y,x,\theta,\eta)\right],
	\qquad
	K(\widetilde\eta',\widetilde\eta)
	=\mathbb{E}_{x;\theta,\eta}\left[r_{\widetilde\eta'}\,r_{\widetilde\eta}\right].
	\label{eq_FD_multiplier}
\end{equation}

\begin{remark}[Existence]
Note that (\ref{eq_FD_multiplier}) is a Fredholm integral equation of the first kind. Its kernel $K$ is positive semi-definite and, in general, compact, so that the inverse operator is unbounded and a solution $\lambda$ need not exist (see, e.g., \citealp{engl1996regularization}). This does not affect the existence of $u_{\infty}^*$ itself. The feasible set in (\ref{eq_FD2}) is a closed affine subspace of $L^2\big(f(\cdot\mid x;\theta,\eta)\big)$ and the objective in (\ref{eq_FD1}) is strictly convex, so the program has a unique solution, given by the projection residual (\ref{eq_fd_proj}). What may fail is only the representation (\ref{eq_FD_solution}) of the correction $u_{\infty}^*-u$ as an integral against a multiplier: the projection belongs to the closed span ${\cal S}_{x;\theta,\eta}$, whereas (\ref{eq_FD_solution}) requires it to lie in the span itself.
\end{remark}

\subsection{Other target parameters\label{App_var_mu}}

We now turn to the targets of Section~\ref{sec_average}. The relevant program is
\begin{align}
	u_{\infty}^*(\cdot,x,\theta,\eta,\mu)=&\, \underset{u^*}{\argmin}\, \mathbb{E}_{x;\theta,\eta}\left[\left(u(y,x,\theta,\eta,\mu)-u^*(y)\right)^2\right],\label{eq_FD_average_min}\\
	& 	\, \text{s.t.}\quad \mathbb{E}_{x;\theta,\widetilde{\eta}}\left[u^*(y)\right]=b_0(x;\theta,\widetilde{\eta},\mu) \text{ for all }\widetilde{\eta}\in\mathcal{H},\label{eq_FD_average_c}
\end{align}
which reduces to (\ref{eq_FD1})--(\ref{eq_FD2}) when $b_0\equiv0$.

The Lagrangian and the first-order condition are as before, so that the solution again takes the form
\begin{equation}
	u_{\infty}^{*}(y,x,\theta,\eta,\mu)=u(y,x,\theta,\eta,\mu)+\int_{{\cal H}}\lambda(\widetilde\eta)\,
	\frac{f(y\mid x;\theta,\widetilde\eta)}{f(y\mid x;\theta,\eta)}\,d\widetilde\eta.
	\label{eq_FD_average_lagrange}
\end{equation}
Imposing the constraints, the multiplier now solves
\begin{equation}
	\int_{{\cal H}} K(\widetilde\eta',\widetilde\eta)\,\lambda(\widetilde\eta)\,d\widetilde\eta
	=\mathbb{E}_{x;\theta,\widetilde\eta'}\left[u(y,x,\theta,\widetilde\eta',\mu)-u(y,x,\theta,\eta,\mu)\right],
	\label{eq_FD_average_kernel}
\end{equation}
with $K$ the Gram kernel of (\ref{eq_FD_multiplier}). The right-hand side measures how much the nuisance argument of the original moment matters, relative to the reference value $\eta$ at which the projection is computed; it vanishes when $u$ does not depend on $\eta$, and (\ref{eq_FD_average_kernel}) reduces to (\ref{eq_FD_multiplier}) when $b_0\equiv0$.

The solution is
\begin{equation}
	u_{\infty}^{*}=g_0+\left(u-\operatorname{Proj}_{x;\theta,\eta}\left(u\,\middle\|\,{\cal S}_{x;\theta,\eta}\right)\right),
	\label{eq_FD_average_sol}
\end{equation}
where $g_0\in{\cal S}_{x;\theta,\eta}$ is the unique element (when it exists) satisfying $\mathbb{E}_{x;\theta,\eta}[g_0\,r_{\widetilde\eta}]=b_0(x;\theta,\widetilde\eta,\mu)$ for all $\widetilde\eta\in\mathcal{H}$. The second term in (\ref{eq_FD_average_sol}) is the fully-orthogonal projection residual of Section \ref{sec_fd}.

\begin{remark}[Existence, other targets]
Unlike in (\ref{eq_FD1})--(\ref{eq_FD2}), the feasible set may now be empty: the zero function satisfies (\ref{eq_FD_average_c}) only when $b_0\equiv0$, and $g_0$ exists only if $b_0$ lies in the range of the map $g\mapsto\left(\widetilde\eta\mapsto\mathbb{E}_{x;\theta,\eta}[g\,r_{\widetilde\eta}]\right)$. Existence of a fully-orthogonal moment for the target is thus a genuine restriction on the model, as Example \ref{example_1} illustrates, and it is distinct from the question of whether the multiplier $\lambda$ in (\ref{eq_FD_average_lagrange}) admits an integral representation.
\end{remark}

\section{Monte Carlo simulations\label{app_simu}}
\setcounter{equation}{0}

This appendix collects the designs and estimators underlying the simulation evidence reported in the text.

\paragraph{Common design.} In each experiment there are $N=1000$ units and $T=5$ periods. The covariate of interest is $x_{it}\sim\text{Bernoulli}(0.5)$ drawn independently across $i$ and $t$, and the individual effects are drawn as $\eta_i\sim{\cal N}(0,1)$, independently of the covariates. The parameter $\theta_0$ moves across a grid from $-2$ to $2$ in increments of $0.25$, with $1000$ replications at each point. Sample splitting follows Section~\ref{sec_rationale}, without cross-fitting: each unit contributes its first $T_{\rm split}=\lfloor T/2\rfloor=2$ periods to the estimation of $\eta_i$, and the estimating equation uses the $T_{\rm est}=T-T_{\rm split}=3$ left-out periods. All estimating equations for $\theta$ are solved on $[-5,5]$, and the first-step equations on $[-10,10]$. An estimate is recorded as missing when its estimating equation has no sign change on that interval; means and percentiles are computed over the remaining replications.

\paragraph{Poisson.} The data are generated from (\ref{eq_poisson}). First-step estimates are obtained by maximizing the split-sample likelihood jointly in $(\theta,\eta_1,\ldots,\eta_N)$. Given $\theta$, the individual effects have the closed form
\begin{equation*}
	\widehat\eta_{i}(\theta)=\ln\Big(\sum_{t\leq T_{\rm split}}y_{it}\Big)-\ln\Big(\sum_{t\leq T_{\rm split}}\exp(x_{it}\theta)\Big),
\end{equation*}
and the split-sample estimate of $\theta$ solves the resulting concentrated score, which is the Poisson conditional-likelihood score. We keep $\widehat\eta_i$ evaluated at the split-sample estimate of $\theta$ and discard the latter; the estimating equations are then solved holding $\widehat\eta_i$ fixed. The closed form is finite provided the unit has a positive count in at least one split-sample period, and units with $\sum_{t\leq T_{\rm split}}y_{it}=0$ are dropped; between about $60\%$ ($\theta_0=-2$) and $93\%$ ($\theta_0=2$) of units are retained. We report the plug-in estimator, solving
$$\sum_i\sum_{t>T_{\rm split}}x_{it}\left(y_{it}-\mu_{it}(\theta,\widehat\eta_i)\right)=0,$$
and the Neyman-orthogonal estimator based on (\ref{eq_poisson_orth}), with the weighted covariate mean computed over the left-out periods. By (\ref{eq_poisson_free}), the latter does not involve $\widehat\eta_i$; it is computed on the same retained units.

\paragraph{Logit.} The data are generated from (\ref{eq_logit}), with the same covariate and effect distributions. Split-sample maximum likelihood again provides $\widehat\eta_i$: for each $\theta$, $\widehat\eta_i(\theta)$ solves $\sum_{t\leq T_{\rm split}}(y_{it}-\pi_{it})=0$ by Newton's method, and the split-sample estimate of $\theta$ solves the resulting profile score. As in the Poisson design, we keep $\widehat\eta_i$ evaluated at the split-sample estimate of $\theta$ and discard the latter. For a unit whose split-sample outcomes are all zero or all one the effect is unbounded, and such units are dropped; roughly $40\%$ of units are retained under this design. All estimators are computed on the retained units. Besides the plug-in estimator, we report the order-$q$ orthogonal estimators for $q=1,2,3$ and the conditional-logit estimator. For $q=1$ the moment is (\ref{eq_logit_orth}). For $q\geq2$ we project the estimation-sample score on the centered powers of $S_i=\sum_{t>T_{\rm split}}(y_{it}-\pi_{it})$, computing the moments $\mathbb{E}[S_i^{j}]$ and $\mathbb{E}[u_iS_i^{j}]$ exactly by enumeration over the $2^{T_{\rm est}}$ outcome configurations of a unit, and solving the resulting $q\times q$ system unit by unit. The conditional-logit score conditions on the choice count $\sum_{t>T_{\rm split}}y_{it}$.

\paragraph{A heterogeneous coefficient.} The third design is based on model (\ref{eq_logit_hetero}), where the individual effect is multiplied by a continuously distributed regressor, and the choice probabilities are
\begin{equation}
	\Pr\left(y_{it}=1\mid x_{it},z_{it},\eta_i\right)=\frac{\exp(x_{it}\theta_0+z_{it}\eta_i)}{1+\exp(x_{it}\theta_0+z_{it}\eta_i)},
	\label{eq_logit_z}
\end{equation}
with $x_{it}\sim\text{Bernoulli}(0.5)$ and $z_{it}\sim\text{Uniform}[1/2,3/2]$ drawn independently across $i$ and $t$, $\eta_i\sim{\cal N}(0,1)$, and no intercept. First-step estimation and the drop rule are as in the logit design, with the individual-effect equation $\sum_{t\leq T_{\rm split}}z_{it}(y_{it}-\pi_{it})=0$. We report the order-$q$ orthogonal estimators for $q=2,3$, obtained by projecting the estimation-sample score on the centered powers of $\sum_{t>T_{\rm split}}z_{it}(y_{it}-\pi_{it})$, with the Gram and cross moments again computed exactly by enumeration over the $2^{T_{\rm est}}$ outcome configurations.

\section{Orthogonal moments in the linear model\label{App_lin_network}}
\setcounter{equation}{0}

This appendix derives the expressions stated in Example \ref{example_3} of Section \ref{sec_average}. We work throughout with $y=x\eta+\sigma\varepsilon$, $\varepsilon\mid x\sim\text{iid }{\cal N}(0,I_n)$, taking $x^{\top}x$ nonsingular and $\sigma^2$ known unless stated otherwise; $\eta$ has dimension $r=d_\eta$. Three matrices will prove useful for the derivations. The elimination matrix $L_r$ satisfies $\operatorname{vech}(C)=L_r\operatorname{vec}(C)$; the commutation matrix $K_r$ satisfies $K_r\operatorname{vec}(C)=\operatorname{vec}(C^{\top})$; and for symmetric $C$ the duplication matrix $D_r$ satisfies $D_r\operatorname{vech}(C)=\operatorname{vec}(C)$. On these, see \citet{magnus1979commutation,magnus1980elimination}.

\subsection{Generalized scores and two lemmas\label{App_lin_lemmas}}

The log-likelihood is $\ell=-\frac{n}{2}\ln(2\pi\sigma^2)-\frac{1}{2\sigma^2}(y-x\eta)^{\top}(y-x\eta)$, with derivatives $\nabla_\eta\ell=x^{\top}(y-x\eta)/\sigma^2$ and $\nabla^2_\eta\ell=-x^{\top}x/\sigma^2$. The first- and second-order generalized scores are therefore
\begin{equation}
	w_1=\frac{1}{\sigma^2}\,x^{\top}(y-x\eta),\qquad
	w_2=\operatorname{vech}(V),\qquad V:=-\frac{1}{\sigma^{2}}\,x^{\top}x+\frac{1}{\sigma^{4}}\,x^{\top}(y-x\eta)(y-x\eta)^{\top}x,
	\label{eq_app_W2}
\end{equation}
where $w_2$ collects the distinct entries of the matrix $\nabla^2_\eta\ell+\nabla_\eta\ell\,\nabla_\eta\ell^{\top}$, and we stack $W_2=(w_1^{\top},w_2^{\top})^{\top}$. This differs from the convention used in the main text where $w_2$ collected all, possibly redundant, entries; the two conventions give the same projection, since the redundant entries are linear combinations of the distinct ones.

\begin{lemma}\label{lem_app_gram}
	$\mathbb{E}_{x;\theta,\eta}\left[w_1w_1^{\top}\right]=\frac{1}{\sigma^{2}}\,x^{\top}x$, $\quad\mathbb{E}_{x;\theta,\eta}\left[w_1w_2^{\top}\right]=0$, and
	\begin{equation}
		\mathbb{E}_{x;\theta,\eta}\left[w_2w_2^{\top}\right]=\frac{1}{\sigma^{4}}\,L_r\left(x^{\top}\otimes x^{\top}\right)\left(I_{n^2}+K_n\right)\left(x\otimes x\right)L_r^{\top}.
		\label{eq_app_Ev2v2}
	\end{equation}
\end{lemma}

\begin{proof}[\bf Proof of Lemma \ref{lem_app_gram}]
Substituting $y-x\eta=\sigma\varepsilon$ gives $w_1=x^{\top}\varepsilon/\sigma$ and $V=x^{\top}(\varepsilon\varepsilon^{\top}-I_n)x/\sigma^{2}$. The first claim is immediate from $\mathbb{E}[\varepsilon\varepsilon^{\top}]=I_n$, and the second holds because all third moments of a centered Gaussian vector vanish. For the third claim, $w_2=\sigma^{-2}L_r(x^{\top}\otimes x^{\top})\operatorname{vec}(\varepsilon\varepsilon^{\top}-I_n)$, and the Gaussian fourth-moment identity, equation (4.3) of \citet{ghazal2000second},
\begin{equation*}
	\mathbb{E}\left[\operatorname{vec}(\varepsilon\varepsilon^{\top})\operatorname{vec}(\varepsilon\varepsilon^{\top})^{\top}\right]=\operatorname{vec}(I_n)\operatorname{vec}(I_n)^{\top}+I_{n^2}+K_n,
\end{equation*}
implies $\mathbb{E}[\operatorname{vec}(\varepsilon\varepsilon^{\top}-I_n)\operatorname{vec}(\varepsilon\varepsilon^{\top}-I_n)^{\top}]=I_{n^2}+K_n$, which yields (\ref{eq_app_Ev2v2}).
\end{proof}

\begin{lemma}\label{lem_app_key}
	For symmetric $r\times r$ matrices $A$ and $B$,
	\begin{equation}
		\operatorname{vech}(A)^{\top}\left[L_r\left(x^{\top}\otimes x^{\top}\right)\left(I_{n^2}+K_n\right)\left(x\otimes x\right)L_r^{\top}\right]^{-1}\operatorname{vech}(B)=\tfrac{1}{2}\operatorname{tr}\left(A\left(x^{\top}x\right)^{-1}B\left(x^{\top}x\right)^{-1}\right).
		\label{eq_app_trace}
	\end{equation}
\end{lemma}

\begin{proof}[\bf Proof of Lemma \ref{lem_app_key}]
We use the identities
\begin{align}
	&D_r=\left(I_{r^2}+K_r\right)L_r^{\top}\left(L_r\left(I_{r^2}+K_r\right)L_r^{\top}\right)^{-1},
	\label{eq_app_P1}\\
	&\left(I_{n^2}+K_n\right)\left(x\otimes x\right)=\left(x\otimes x\right)\left(I_{r^2}+K_r\right),
	\label{eq_app_P2}\\
	&D_rL_r\left(I_{r^2}+K_r\right)=I_{r^2}+K_r,
	\label{eq_app_P3}\\
	&K_rD_r=D_r,
	\label{eq_app_P4}
\end{align}
from \citet{magnus1980elimination}. By (\ref{eq_app_P2}), the matrix inside the inverse equals $L_r(x^{\top}x\otimes x^{\top}x)(I_{r^2}+K_r)L_r^{\top}$, and combining (\ref{eq_app_P1})--(\ref{eq_app_P4}) shows that its inverse is $\frac{1}{2}D_r^{\top}\left((x^{\top}x)^{-1}\otimes(x^{\top}x)^{-1}\right)D_r$. Since $D_r\operatorname{vech}(A)=\operatorname{vec}(A)$ for symmetric $A$, the left-hand side of (\ref{eq_app_trace}) equals $$\frac{1}{2}\operatorname{vec}(A)^{\top}\left((x^{\top}x)^{-1}\otimes(x^{\top}x)^{-1}\right)\operatorname{vec}(B)=\frac{1}{2}\operatorname{tr}(A(x^{\top}x)^{-1}B(x^{\top}x)^{-1}).$$
\end{proof}

\subsection{Variance components\label{App_lin_proof}}

Here $\mu_0=\eta_0^{\top}Q\eta_0$ and the base moment is $u=\eta^{\top}Q\eta-\mu$, which does not involve $y$, so that $\mathbb{E}[u\,W_2^{\top}]=u\,\mathbb{E}[W_2^{\top}]=0$ and the construction (\ref{eq_qorth_explicit_target}) reduces to
\begin{equation*}
	u^*_2=u+b_2^{\top}\left(\mathbb{E}\left[W_2W_2^{\top}\right]\right)^{-1}W_2,
\end{equation*}
where $b_2=\big((2Q\eta)^{\top},\operatorname{vech}(2Q)^{\top}\big)^{\top}$ is the column vector collecting the sensitivities of Section \ref{sec_average} in the coordinates of $W_2$: the gradient $\nabla_\eta(\eta^{\top}Q\eta)=2Q\eta$ and, in the vech convention of (\ref{eq_app_W2}), the second-derivative block $\operatorname{vech}(2Q)$. By Lemma~\ref{lem_app_gram} the Gram matrix is block diagonal, so the two blocks contribute separately. The first-order block contributes
\begin{equation*}
	(2Q\eta)^{\top}\left(\frac{x^{\top}x}{\sigma^{2}}\right)^{-1}\frac{x^{\top}(y-x\eta)}{\sigma^{2}}=2\,\eta^{\top}Q(x^{\top}x)^{-1}x^{\top}(y-x\eta),
\end{equation*}
and truncating the construction after this block yields $u_1^*$ in (\ref{eq_akm_u1}). The second-order block contributes, by Lemma~\ref{lem_app_key} with $A=2Q$ and $B=\sigma^4 V$,
\begin{align*}
	&\sigma^{4}\operatorname{vech}(2Q)^{\top}\left[L_r(x^{\top}\otimes x^{\top})(I_{n^2}+K_n)(x\otimes x)L_r^{\top}\right]^{-1}\operatorname{vech}(V)
	\\&=\sigma^{4}\operatorname{tr}\left(Q(x^{\top}x)^{-1}V(x^{\top}x)^{-1}\right),
\end{align*}
and substituting the definition of $V$,
\begin{align*}
	&\sigma^{4}\operatorname{tr}\left(Q(x^{\top}x)^{-1}V(x^{\top}x)^{-1}\right)\\&=(y-x\eta)^{\top}x(x^{\top}x)^{-1}Q(x^{\top}x)^{-1}x^{\top}(y-x\eta)-\sigma^{2}\operatorname{tr}\left(Q(x^{\top}x)^{-1}\right).
\end{align*}
Adding the three pieces and expanding $(y-x\eta)^{\top}x=y^{\top}x-\eta^{\top}x^{\top}x$, the terms involving $\eta$ cancel,
\begin{align*}
	&\eta^{\top}Q\eta+2\,\eta^{\top}Q(x^{\top}x)^{-1}x^{\top}(y-x\eta)+(y-x\eta)^{\top}x(x^{\top}x)^{-1}Q(x^{\top}x)^{-1}x^{\top}(y-x\eta)\\&=y^{\top}x(x^{\top}x)^{-1}Q(x^{\top}x)^{-1}x^{\top}y,
\end{align*}
which gives $u^*_2$ in (\ref{eq_akm_u2}). As a check, substituting $x^{\top}y=x^{\top}x\eta+\sigma x^{\top}\varepsilon$ shows directly that $\mathbb{E}[y^{\top}x(x^{\top}x)^{-1}Q(x^{\top}x)^{-1}x^{\top}y]=\eta^{\top}Q\eta+\sigma^2\operatorname{tr}(Q(x^{\top}x)^{-1})$ for every $\eta$, confirming that $u^*_2$ is fully orthogonal, and, in this case, exactly unbiased.

\subsection{The variance parameter\label{App_lin_sigma}}

We now derive the expression given in footnote \ref{ftnote_ex3}, taking $\theta=\sigma^2$ as the target and $u=\nabla_{\sigma^2}\ell=-\frac{n}{2\sigma^2}+\frac{1}{2\sigma^4}(y-x\eta)^{\top}(y-x\eta)$. Being a score, $u$ is unbiased at all parameter values, so that all sensitivities vanish and the construction of Section \ref{sec_higher_constr} applies, with $u^*_2=u-\mathbb{E}[u\,W_2^{\top}]\left(\mathbb{E}[W_2W_2^{\top}]\right)^{-1}W_2$.

Writing $u=(\varepsilon^{\top}\varepsilon-n)/(2\sigma^2)$, the first-order block vanishes because $\mathbb{E}[u\,w_1^{\top}]$ involves third moments of $\varepsilon$. For the second-order block, $\mathbb{E}[\varepsilon_i\varepsilon_j\varepsilon_k\varepsilon_l]=\delta_{ij}\delta_{kl}+\delta_{ik}\delta_{jl}+\delta_{il}\delta_{jk}$ gives $\mathbb{E}[(\varepsilon^{\top}\varepsilon)\varepsilon\varepsilon^{\top}]=(n+2)I_n$, hence $\mathbb{E}[(\varepsilon^{\top}\varepsilon-n)\,\varepsilon\varepsilon^{\top}]=2I_n$ and
\begin{equation*}
	\mathbb{E}\left[u\,V\right]=\frac{1}{2\sigma^{4}}\,x^{\top}\,\mathbb{E}\left[(\varepsilon^{\top}\varepsilon-n)\left(\varepsilon\varepsilon^{\top}-I_n\right)\right]x=\frac{1}{\sigma^{4}}\,x^{\top}x .
\end{equation*}
Applying Lemma~\ref{lem_app_key} with $A=x^{\top}x/\sigma^{4}$ and $B=\sigma^{4}V$, the second-order block contributes $\tfrac{1}{2}\operatorname{tr}\left(V(x^{\top}x)^{-1}\right)$, so that
\begin{equation*}
	u^*_2=u-\tfrac{1}{2}\operatorname{tr}\left(V(x^{\top}x)^{-1}\right).
\end{equation*}
Finally, writing $P_x=x(x^{\top}x)^{-1}x^{\top}$ and $M_x=I_n-P_x$, the definition of $V$ gives $\operatorname{tr}(V(x^{\top}x)^{-1})=\left(\varepsilon^{\top}P_x\varepsilon-d_\eta\right)/\sigma^{2}$, and since $M_xx=0$ implies $y^{\top}M_xy=\sigma^{2}\varepsilon^{\top}M_x\varepsilon$,
\begin{equation*}
	u^*_2=\frac{\varepsilon^{\top}\varepsilon-n}{2\sigma^{2}}-\frac{\varepsilon^{\top}P_x\varepsilon-d_\eta}{2\sigma^{2}}
	=\frac{\varepsilon^{\top}M_x\varepsilon-\left(n-d_\eta\right)}{2\sigma^{2}}
	=\frac{1}{2\sigma^{4}}\left(y^{\top}M_xy-\left(n-d_\eta\right)\sigma^{2}\right),
\end{equation*}
as claimed. It is free of $\eta$ and has mean zero at every value of it, since $\mathbb{E}[\varepsilon^{\top}M_x\varepsilon]=\operatorname{tr}(M_x)=n-d_\eta$.

\end{document}